\documentclass[a4paper,american]{lipics-v2019}

\usepackage{enumerate}
\usepackage{amsopn}
\usepackage{amsmath,tabu}
\usepackage{stackengine}
\usepackage{amsfonts}
\usepackage{complexity}
\usepackage{mathtools}
\usepackage{tabularx}
\usepackage{environ}
\usepackage{hyperref}

\newcommand{\lcs}{\mathcal{S}}
\newcommand{\initmem}{a^0}
\newcommand{\Ops}[1]{\mathit{Op}(#1)}
\newcommand{\Domain}{\mathit{D}}
\newcommand{\wt}[1]{! #1}
\newcommand{\rd}[1]{? #1}

\newcommand{\pc}{\mathit{pc}}
\newcommand{\Conf}{\mathit{CF}}
\newcommand{\WritesSCC}{\Writes_{\mathit{SCC}}}

\newcommand{\ld} {{\it ld}}
\newcommand{\ct} {{\it ct}}
\newcommand{\lp} {{\it lp}}
\newcommand{\idx} {{\bf idx}}
\newcommand{\white}{{ \it \#\!Wt }}

\newcommand{\ord}{\mathit{ord}}
\newcommand{\init}{\mathit{init}}
\newcommand{\Wit}{\mathit{Wit}}
\newcommand{\SWit}{\Wit^{\mathit{sh}}}
\newcommand{\Trace}{\mathit{Trace}}
\newcommand{\Expr} {\mathit{Expr}}
\newcommand{\FullExpr}{\mathit{FullExpr}}
\newcommand{\Shrink}{\mathit{Shrink}}
\newcommand{\SValid}{\Valid^{\mathit{sh}}}

\DeclareMathOperator{\Inf}{Inf}
\DeclareMathOperator{\Writes}{Writes}

\DeclareMathOperator{\LValid}{LValid}
\DeclareMathOperator{\CValid}{CValid}
\DeclareMathOperator{\Valid}{Valid}
\DeclareMathOperator{\Loop}{Loop}
\DeclareMathOperator{\Ord}{Ord}

\newcommand{\sat}{\mathit{sat}}
\newcommand{\cop}{\mathit{cc}}
\newcommand{\restrictionGraph}[2]{G_{#1}(#2)}
\newcommand{\decomp}[2]{\mathit{SCCdcmp}_{#1}(#2)}

\newcommand{\inc}{\mathit{inc}}
\newcommand{\card}{\mathit{card}}
\newcommand{\pre}{\mathit{pre}}
\newcommand{\po}{\mathit{po}}

\DeclareMathOperator{\IF}{IF}
\DeclareMathOperator{\FW}{FW}
\DeclareMathOperator{\States}{States}
\DeclareMathOperator{\cyc}{cycle}
\DeclareMathOperator{\Gen}{Gen}

\newcommand{\setcon}[1]{ \lbrace #1 \rbrace }
\newcommand{\Setcon}[2]{ \lbrace #1 \mid #2 \rbrace }
\newcommand{\Naturals}{\mathbb{N}}
\newcommand{\Powerset}{\mathcal{P}}

\newcommand{\langu}{\mathcal{L}}
\newcommand{\TRUE}{\mathit{true}}

\newcommand{\abs}[1]{|#1|}
\newcommand{\proj}[1]{\pi_{#1}}

\newcommand{\bigO}{\mathcal{O}}
\newcommand{\bigOS}{\bigO^*}
\newcommand{\LCL}{\ComplexityFont{LCL}}
\newcommand{\LCR}{\ComplexityFont{LCR}}
\newcommand{\CYC}{\ComplexityFont{CYC}}
\newcommand{\ETH}{\ComplexityFont{ETH}}
\newcommand{\kSAT}[1]{#1\text{-}\SAT}
\newcommand{\NEXPTIME}{\ComplexityFont{NEXPTIME}}

\newcommand{\reach}{\mathit{Reach}(\sizeL,\sizeD,\sizeC)}
\newcommand{\cycle}{\mathit{Cycle}(\sizeL,\sizeD,\sizeC)}
\newcommand{\reachEmph}{\mathit{Reach}(\sizeL,\sizeD,\sizeC \,)}
\newcommand{\cycleEmph}{\mathit{Cycle}(\sizeL,\sizeD,\sizeC \,)}

\newcommand{\evaltime}{\sizeD \cdot (\sizeC^{2} + \sizeL^{2} \cdot \sizeD^{2})}
\newcommand{\evaltimeEmph}{\sizeD \cdot (\sizeC^{\,2} + \sizeL^{\,2} \cdot \sizeD^{\,2})}

\newcommand{\fixpointtime}{\sizeD^{2} \cdot (\sizeC^{2} + \sizeL^{2} \cdot \sizeD^{2})}
\newcommand{\fixpointtimeEmph}{\sizeD^{\,2} \cdot (\sizeC^{\,2} + \sizeL^{\,2} \cdot \sizeD^{\,2})}

\newcommand{\contributortimeEmph}{2^{\sizeC} \cdot \sizeL \cdot \sizeD^{\,2} \cdot (\sizeL \cdot \sizeC^{\,4} + \sizeD \cdot \sizeC^{\,2} + \sizeL^{\,2} \cdot \sizeD^{\,3})}

\newcommand{\contributortimeReach}{2^\sizeC \cdot \sizeC^4 \cdot \sizeL^2 \cdot \sizeD^2}
\newcommand{\contributortimeReachEmph}{2^\sizeC \cdot \sizeC^{\,4} \cdot \sizeL^{\,2} \cdot \sizeD^{\,2}}

\newcommand{\leaderdomaintimeEmph}{(\sizeL + \sizeD)^{\bigO(\sizeL + \sizeD)}}

\newcommand{\sizeL}{\texttt{L}}
\newcommand{\sizeD}{\texttt{D}}
\newcommand{\sizeC}{\texttt{C}}

\stackMath
\newcommand\XXrightarrow[2]{
	\raisebox{-.75pt}{
		\ensuremath{
			\smash{
				\mathrel{%
					\setbox2=\hbox{\stackon{\scriptstyle#1}{\scriptstyle#1}}%
					\stackon[-2pt]{%
						\xrightarrow{\makebox[\dimexpr\wd2\relax]{}}%
					}{%
					\scriptstyle#1\,%
					}%
				}_{#2}
			}
		}
	}
}

\stackMath
\newcommand\XXrightarrowLow[2]{
	\raisebox{-.75pt}{
		\ensuremath{
			\smash{
				\mathrel{%
					\setbox2=\hbox{\stackon{\scriptstyle#1}{\scriptstyle#1}}%
					\stackon[-3.5pt]{%
						\xrightarrow{\makebox[\dimexpr\wd2\relax]{}}%
					}{%
					\scriptstyle#1\,%
					}%
				}_{#2}
			}
		}
	}
}

\stackMath
\newcommand\xxrightarrow[1]{
	\raisebox{-.75pt}{
		\ensuremath{
			\smash{
				\mathrel{%
					\setbox2=\hbox{\stackon{\scriptstyle#1}{\scriptstyle#1}}%
					\stackon[-2pt]{%
						\xrightarrow{\makebox[\dimexpr\wd2\relax]{}}%
					}{%
					\scriptstyle#1\,%
					}%
				}
			}
		}
	}
}

\stackMath
\newcommand\xxrightarrowLow[1]{
	\raisebox{-.75pt}{
		\ensuremath{
			\smash{
				\mathrel{%
					\setbox2=\hbox{\stackon{\scriptstyle#1}{\scriptstyle#1}}%
					\stackon[-3.5pt]{%
						\xrightarrow{\makebox[\dimexpr\wd2\relax]{}}%
					}{%
					\scriptstyle#1\,%
					}%
				}
			}
		}
	}
}

\makeatletter
\newcommand{\problemtitle}[1]{\gdef\@problemtitle{#1}}
\newcommand{\problemshort}[1]{\gdef\@problemshort{#1}}
\newcommand{\probleminput}[1]{\gdef\@probleminput{#1}}
\newcommand{\problemparameter}[1]{\gdef\@problemparameter{#1}}
\newcommand{\problemquestion}[1]{\gdef\@problemquestion{#1}}

\NewEnviron{myproblem}{
	\problemtitle{}
	\problemshort{}
	\probleminput{}
	\problemquestion{}
	
	\BODY
	\par\addvspace{.5\baselineskip}
	\noindent
	\framebox[\textwidth]{
		\begin{tabularx}{\textwidth}{@{\hspace{\parindent}} l X c}
			\multicolumn{2}{@{\hspace{\parindent}}l}{ \normalsize \emph{\@problemtitle} \@problemshort} \\
			\normalsize \textbf{Input:} & \normalsize \@probleminput \\
			\normalsize \textbf{Question:} & \normalsize \@problemquestion
		\end{tabularx}
	}
	\par\addvspace{.5\baselineskip}
}

\title{Complexity of Liveness in Parameterized Systems}

\author{Peter Chini}{TU Braunschweig}{p.chini@tu-braunschweig.de}{}{}
\author{Roland Meyer}{TU Braunschweig}{roland.meyer@tu-braunschweig.de}{}{}
\author{Prakash Saivasan}{TU Braunschweig}{p.saivasan@tu-braunschweig.de}{}{}

\authorrunning{P. Chini, R. Meyer, and P. Saivasan}

\Copyright{Peter Chini, Roland Meyer, and Prakash Saivasan}

\ccsdesc[500]{Theory of computation~Formal languages and automata theory}
\ccsdesc[500]{Theory of computation~Problems, reductions and completeness}

\keywords{Liveness Verification, Fine-Grained Complexity, Parameterized Systems.}

\nolinenumbers

\begin{document}

\maketitle

\begin{abstract}
	We investigate the fine-grained complexity of liveness verification for leader contributor systems.
	These consist of a designated leader thread and an arbitrary number of identical contributor threads communicating via a shared memory.
	The liveness verification problem asks whether there is an infinite computation of the system in which the leader reaches a final state infinitely often.
	Like its reachability counterpart, the problem is known to be $\NP$-complete.
	Our results show that, even from a fine-grained point of view, the complexities differ only by a polynomial factor.
	
	Liveness verification decomposes into reachability and cycle detection.
	We present a fixed point iteration solving the latter in polynomial time.
	For reachability, we reconsider the two standard parameterizations.
	When parameterized by the number of states of the leader $\sizeL$ and the size of the data domain $\sizeD$, we show an $(\sizeL + \sizeD)^{\bigO(\sizeL + \sizeD)}$-time algorithm.
	It improves on a previous algorithm, thereby settling an open problem.
	When parameterized by the number of states of the contributor $\sizeC$, we reuse an $\bigOS(2^\sizeC)$-time algorithm.
	We show how to connect both algorithms with the cycle detection to obtain algorithms for liveness verification.
	The running times of the composed algorithms match those of reachability, proving that the fine-grained lower bounds for liveness verification are met.
\end{abstract}

\section{Introduction}
\label{Section:Introduction}

We study the fine-grained complexity of liveness verification for parameterized systems formulated in the leader contributor model. 
The model~\cite{Hague2011,Esparza2013} assumes a distinguished leader thread interacting (via a shared memory) with a finite but arbitrary number of indistinguishable contributor threads.
The liveness verification problem~\cite{Esparza2015} asks whether there is an infinite computation of the system in which the leader visits a set of final states infinitely often. 
Fine-grained complexity~\cite{Downey2013,Cygan2015} studies the impact of parameters
associated with an algorithmic problem on the problem's complexity like the influence of the contributor size on the complexity of liveness verification. 
The goal is to develop deterministic algorithms that are provably optimal.
We elaborate on the three ingredients of our study.

The leader contributor model has attracted considerable attention~\cite{Hague2011,Esparza2013,Esparza2015,Muscholl2015,Esparza2016,Fortin2017,Chini2018}.
From a modeling point of view, a variety of systems can be formulated as anonymous entities interacting with a central authority, examples being client-server applications, resource-management systems, and distributed protocols on wireless sensor networks.
From an algorithmic point of view, the model has led to positive surprises.
Hague~\cite{Hague2011} proved decidability of reachability even in a setting where the system components are pushdown automata. 
La Torre et al.~\cite{Muscholl2015} generalized the result to any class of components that satisfies mild assumptions, the most crucial of which being computability of downward closures.
As for the complexity, Esparza et al.~\cite{Esparza2013,Esparza2016} proved $\PSPACE$-completeness for Hague's model and $\NP$-completeness in the setting where the components are given by finite-state automata.
The liveness problem was first studied in~\cite{Esparza2015}.
Interestingly, liveness has the same complexity as reachability, it is $\NP$-complete for finite-state systems.
Fortin et al.~\cite{Fortin2017} generalized the study to LTL-definable properties and gave conditions for~$\NEXPTIME$-completeness.

Fine-grained complexity is a field within parameterized complexity~\cite{Downey2013,Cygan2015}. 
Parameterized complexity intends to explain the following gap between theory and practice that is observed throughout algorithmics.  
Despite a high worst-case complexity, tools may have an easy time solving a problem. 
Parameterized complexity argues that measuring the complexity of a problem in terms of the size of the input, typically denoted by $n$, is too rough.
One should consider further parameters $k$ that capture the shape of the input or the solution sought. 
Then the gap is due to the fact that tools implement an algorithm running in time $f(k)\cdot \mathit{poly}(n)$. 
Here, $f$ may be an exponential, but it only depends on the parameter, and that parameter is small in practice. 
Problems solvable by such an algorithm are called fixed-parameter tractable and belong to the complexity class $\FPT$. 
Fine-grained complexity is the study of the precise function $f$ that is needed, via upper and lower bound arguments. 

The fine-grained complexity of the reachability problem for the leader contributor model was studied in our previous work~\cite{Chini2018}. 
We assumed that the components are finite state and considered two parameterizations.
When parameterized by the size of the contributors~$\sizeC$, we showed that reachability can be solved in time~$\bigOS(2^{\sizeC})$.
The notation $\bigOS$ suppresses polynomial factors in the running time.
Interestingly, this is the best one can hope for.
An algorithm with a subexponential dependence on $\sizeC$, to be precise an algorithm running in time $2^{o(\sizeC)}$, would contradict the so-called exponential time hypothesis ($\ETH$).
The $\ETH$ \cite{Impagliazzo2001} is a standard hardness assumption in parametrized complexity that is used to derive relative lower bounds.
The second parameterization is by the size of the leader $\sizeL$ and the size of the data domain $\sizeD$.
We gave an algorithm running in time $(\sizeL\sizeD)^{\bigO(\sizeL\sizeD)}$.
Interestingly, the lower bound is only $2^{o((\sizeL+\sizeD) \cdot \log(\sizeL+\sizeD))}$.
Being away a quadratic factor in the exponent means a substantial gap for a deterministic algorithm.

In the present paper, we study the fine-grained complexity of the liveness verification problem.
We assume finite-state components and consider the same parameterization as for reachability. 
The surprise is in the parameterization by $\sizeL$ and $\sizeD$. 
We give an algorithm running in time $(\sizeL+\sizeD)^{\bigO(\sizeL+\sizeD)}$. This matches the lower bound and closes the gap for reachability.
When parameterized by the size of the contributors, we obtain an $\bigOS(2^\sizeC)$ algorithm. 

To explain the algorithms, note that a live computation decomposes into a prefix and an accepting cycle.
Finding prefixes is a matter of reachability.
We show how to combine reachability algorithms with a cycle detection to obtain algorithms that find live computations.
The resulting algorithms will run in time $\bigO(\reach \cdot \cycle)$ where $\reach$ denotes the running time of the invoked reachability algorithm and $\cycle$ that of the cycle detection.
This result allows for considering reachability and cycle detection separately.

Our first main contribution is an algorithm for reachability when $\sizeL$ and $\sizeD$ are given as parameters.
It runs in time $\leaderdomaintimeEmph$ and significantly improves upon the $(\sizeL \sizeD)^{\bigO(\sizeL \sizeD)}$-time algorithm from \cite{Chini2018}.
Moreover, it is optimal in the fine-grained sense.
It closes the gap between upper and lower bound.
The algorithm works over sketches of computations.
A sketch is valid if there is an actual computation corresponding to it.
In~\cite{Chini2018}, we performed a single validity check for each sketch.
Here, we show that valid sketches can be build up inductively from small sketches.
To this end, we interleave validity checks with compression phases.
Our algorithm is a dynamic programming on small sketches, exploiting the inductive approach.

Our second main result is an algorithm for detecting cycles.
We show that the problem is actually solvable in polynomial time.
Technically, we employ a characterization of cycles via (certain) SCC decompositions of the contributor automaton.
These decompositions can be computed by a fixed point iteration invoking Tarjan's algorithm \cite{Tarjan1972} in polynomial time.

Since $\cycle$ is polynomial, liveness has the same complexity as reachability also in the fine-grained sense.
With the above result, we obtain the mentioned algorithms for liveness by composing the reachability algorithms with the cycle detection.

\subparagraph*{Related Work.}
\label{Section:RelatedWork}

The parameterized complexity has also been studied for other verification problems.
Farzan and Madhusudan~\cite{Farzan2009} consider the problem of predicting atomicity violations. 
Depending on the synchronization, they obtain an efficient fine-grained algorithm resp. prove an $\FPT$-algorithm unlikely.
In \cite{Farzan2016}, the authors give an efficient (fine-grained) algorithm for the problem of checking TSO serializability.
In~\cite{Meyer2017}, we studied the fine-grained complexity of bounded context switching~\cite{Qadeer2005}, including lower bounds on the complexity.
In~\cite{Chini2018}, we gave a parameterized analysis of the bounded write-stage restriction, a generalization of bounded context switching~\cite{Saivasan2014}.
The problem turns out to be hard for different parameterizations, and has a large number of hard instances.
In a series of papers~\cite{Fernau2016,Fernau2015,Wareham2000}, Fernau et al. studied $\FPT$-algorithms for problems from automata theory.

Related to leader contributor systems are broadcast networks (ad-hoc networks) \cite{Singh2009,Delzanno2010}.
These consist of an arbitrary number of finite-state contributors that communicate via message passing.
There is no leader.
This has an impact on the complexity of safety \cite{Delzanno2012,Fournier2015} and liveness \cite{Saivasan2019,Bertrand2014} verification, which drops from $\NP$ (leader contributor systems) to $\P$.

More broadly, the verification of parameterized systems is an active field of research~\cite{Bloem2015}. 
Prominent approaches are well-structuredness arguments~\cite{Abdulla1993,Finkel2001} and cut-off results~\cite{German1992}.
Well-structuredness means the transition relation is monotonic wrt. a well-quasi ordering on the configurations, a combination that leads to surprising decidability results. 
A cut-off is a bound on the size of system instances such that correctness of the bounded instances entails correctness of all instances.  
Our algorithm uses different techniques.
We give a reduction from liveness to reachability combined with a polynomial-time cycle check.
Reductions from liveness to reachability or safety are recently gaining popularity in verification~\cite{Konnov2017,Padon2018,Hague2018}. 
For reachability, we then rely on techniques from parameterized complexity~\cite{Downey2013,Cygan2015}, namely identifying combinatorial objects to iterate over and dynamic programming.

\section{Leader Contributor Systems and the Liveness Problem}
\label{Section:Preliminaries}

We introduce leader contributor systems and the leader contributor liveness problem of interest following \cite{Hague2011,Esparza2013,Esparza2015}.
Moreover, we give a short introduction to fine-grained complexity.
For standard textbooks, we refer to \cite{Fomin2010,Cygan2015,Downey2013}.

\subparagraph*{Leader Contributor Systems.}
\label{Section:LeaderContributorSystems}

A \emph{leader contributor system} consists of a designated leader thread communicating with a number of identical contributor threads via a shared memory. 
Formally, the system is a tuple $\lcs = (\Domain, \initmem, P_L, P_C)$ where $D$ is the finite domain of the shared memory and $\initmem \in \Domain$ is the initial memory value.
The leader $P_L$ and the contributor $P_C$ are abstractions of concrete threads making visible the interaction with the memory.
They are defined as finite state automata over the alphabet $\Ops{\Domain} = \Setcon{\wt{a},\rd{a}}{a \in \Domain}$ of memory operations.
Here, $\wt{a}$ denotes a write of $a$ to the memory, $\rd{a}$ denotes a read of $a$.
The leader is given by the tuple $P_L = (\Ops{\Domain}, Q_L, q^0_L, \delta_L)$ where $Q_L$ is the set of states, $q^0_L \in Q_L$ is the initial state, and $\delta_L \subseteq Q_L \times (\Ops{\Domain} \cup \setcon{\varepsilon}) \times Q_L$ is the transition relation.
We extend the relation to words in $\Ops{\Domain}^*$ and usually write $q \XXrightarrow{w}{L} q'$ for $(q,w,q') \in \delta_L$.
The contributor is defined similarly, by $P_C = (\Ops{\Domain}, Q_C, q^0_C, \delta_C)$.

The possible interactions of a thread with the memory depend on the current memory value and the internal state of the thread.
To keep track of this information, we use \emph{configurations}.
These are tuples of the form $(q,a,\pc) \in \Conf^t = Q_L \times \Domain \times Q^t_C$.
Here, $\pc$ is a vector storing the current state of each contributor, and there are $t\in\Naturals$ contributors participating in the computation.
The number of participating contributors can be arbitrary, but will be fixed throughout the computation.
Therefore, the set of all configurations is given by $\Conf = \bigcup_{t \in \Naturals} \Conf^t$.
A configuration is called \emph{initial} if it is of the form $(q^0_L,\initmem,\pc^0)$ where $\pc^0(i) = q^0_C$ for each $i \in [1..t]$.
We use projections to access the components of a configuration.
Let $\proj{L}$ and $\proj{\Domain}$ denote the projections to the leader state resp. the memory content, $\proj{L}((q,a,\pc)) = q$ and $\proj{\Domain}((q,a,\pc)) = a$.
The map $\proj{C}$ projects a configuration to the set of contributor states present in $\pc$, $\proj{C}((q,a,\pc)) = \Setcon{\pc(i)}{ i \in [1..t]}$.

The current configuration of $\lcs$ may change due to an interaction with the memory or an internal transition.
We capture such changes by a labeled transition relation among configurations, $\rightarrow\ \subseteq \Conf \times (\Ops{\Domain} \cup \setcon{\varepsilon}) \times \Conf$.
It contains transitions induced by the leader and by the contributor.
We focus on the former.
If there is a write $q \XXrightarrow{\wt{b}}{L} q'$ of the leader, we get $(q,a,\pc) \xxrightarrow{\wt{b}} (q',b,\pc)$.
Similarly, a read $q \XXrightarrow{\rd{a}}{L} q'$ induces $(q,a,\pc) \xxrightarrow{\rd{a}} (q',a,\pc)$.
Note that the current memory value has to match the read symbol.
An internal transition $q \XXrightarrow{\varepsilon}{L} q'$ yields $(q,a,\pc) \xxrightarrow{\varepsilon} (q',a,\pc)$.
For the transitions induced by the contributors,
let $\pc(i) = p$ and $\pc' = \pc[i = p']$, meaning $\pc'(i) = p'$ and $\pc'$ coincides with $\pc$ in all other components. 
A transition $p \XXrightarrowLow{\wt{b} / \rd{a} / \varepsilon}{C} p'$ yields $(q,a,\pc) \xxrightarrowLow{\wt{b} / \rd{a} / \varepsilon} (q,b/a,\pc')$, like for the leader.
Note that transitions are only defined among configurations involving the same number of contributors.
It is convenient to assume that the leader never writes $\wt{a}$ and immediately reads $\rd{a}$ again.
In this case, we could replace the corresponding read transition by $\varepsilon$.

The transition relation $\rightarrow$ is generalized to words, denoted by $c \xxrightarrow{w} c'$ with $w \in \Ops{\Domain}^*$.
We call such a sequence a \emph{computation} of $\lcs$.
We also write $c \rightarrow^* c'$ if there is a word $w$ with $c \xxrightarrow{w} c'$, and $c \rightarrow^+ c'$ if $w$ has length at least~$1$.
An \emph{infinite computation} is a sequence $\sigma = c^0 \rightarrow c^1 \rightarrow \dots$ of infinitely many transitions.
We call it \emph{initialized} if $c^0$ is an initial configuration.
Since $\sigma$ involves infinitely many configurations but the set $Q_L$ is finite, there are states of the leader that occur infinitely often along the computation.
We denote the set of these states by $\Inf(\sigma) = \Setcon{q \in Q_L}{\exists^\infty \, i : q = \proj{L}(c^i)}$.

\subparagraph*{Leader Contributor Liveness.}
\label{Section:LeaderContributorLiveness}

The \emph{leader contributor liveness problem} is the task of deciding whether the leader satisfies a liveness specification while interacting with a number of contributors.
Formally, given a leader contributor system $\lcs = (\Domain, \initmem, P_L, P_C)$ and a set of final states $F \subseteq Q_L$ encoding the specification, the problem asks whether there is an initialized infinite computation $\sigma$ such that the leader visits $F$ infinitely often along $\sigma$.
Since $F$ is finite, this is equivalent to $\Inf(\sigma) \cap F \neq \emptyset$.
In this case, $\sigma$ is called a \emph{live computation}.
\begin{myproblem}
	\problemtitle{Leader Contributor Liveness}
	\problemshort{($\LCL$)}
	\probleminput{A leader contributor system $\lcs = (\Domain, \initmem, P_L, P_C)$ and final states $F \subseteq Q_L$.}
	\problemquestion{Is there an infinite initialized computation $\sigma$ such that $\Inf(\sigma) \cap F \neq \emptyset$?}
\end{myproblem}

\subparagraph*{Fine-Grained Complexity.}
\label{Section:FPTbasic}

The problem $\LCL$ is known to be $\NP$-complete \cite{Esparza2015}.
Despite its hardness, it may still admit efficient deterministic algorithms the running times of which depend exponentially only on certain parameters.
To find parameters that allow for the construction of such algorithms, one examines the \emph{parameterized complexity} of $\LCL$. 
Note that the name does not refer to parameterized systems.
It stems from measuring the complexity not only in the size of the input but also in the mentioned parameters.

Let $\Sigma$ be an alphabet.
Unlike in classical complexity theory where we consider problems over $\Sigma^*$, a \emph{parameterized problem} $P$ is a subset of $\Sigma^* \times \Naturals$.
Inputs to $P$ are pairs $(x,k)$ with the second component $k$ being referred to as the \emph{parameter}.
Problem $P$ is called \emph{fixed-parameter tractable} if it admits a deterministic algorithm deciding membership in $P$ for pairs $(x,k)$ in time $f(k) \cdot \abs{x}^{\bigO(1)}$.
Here, $f$ is a computable function that only depends on $k$.
Since $f$ usually dominates the polynomial, the running time of the algorithm is denoted by $\bigOS(f(k))$.

While finding an upper bound for the function $f$ amounts to coming up with an efficient algorithm, lower bounds on $f$ are obtained relative to hardness assumptions.
One of the standard assumptions is the \emph{exponential time hypothesis} ($\ETH$) \cite{Impagliazzo2001}.
It asserts that $\kSAT{3}$ cannot be solved in time $2^{o(n)}$ where $n$ is the number of variables in the input formula.
The lower bound is transported to the problem of interest via a reduction from $\kSAT{3}$.
Then, $f$ cannot drop below a certain bound unless $\ETH$ fails.
It is a task of \emph{fine-grained complexity} to find the \emph{optimal} function $f$, where upper and lower bound match.

We conduct fine-grained complexity analyses for two parameterizations of $\LCL$.
First, we consider $\LCL(\sizeL,\sizeD)$, the parameterization by the number of states in the leader $\sizeL$ and the size of the data domain~$\sizeD$.
We show an $(\sizeL + \sizeD)^{\bigO(\sizeL + \sizeD)}$-time algorithm, matching the lower bound for $\LCL$ from~\cite{Chini2018}.
The second parameterization $\LCL(\sizeC)$ is by the number of states of the contributor $\sizeC$.
We give an algorithm running in time $\bigOS(2^{\sizeC})$. 
It also matches the known lower bound \cite{Chini2018}.
Therefore, both algorithms are optimal in the fine-grained sense.
The parameterizations $\LCL(\sizeL)$ and $\LCL(\sizeD)$ are unlikely to be fixed-parameter tractable.
These problems are hard for $\W[1]$, a complexity class comprising intractable problems~\cite{Chini2018}.

\section{Dividing Liveness along Interfaces}
\label{Section:Interfaces}

A live computation naturally decomposes into a prefix and a cycle.
This means that solving $\LCL$ amounts to finding both, a prefix computation and a cyclic computation.
However, we need to guarantee that the computations can be linked.
The prefix should lead to a configuration that the cycle loops on.
Since there are infinitely many configurations, we introduce the finite domain of interfaces.
An interface abstracts a configuration to its leader state, memory value, and set of contributor states.
Hence, an interface can be seen as a summary of those configurations that are suitable for linking prefix and cycle.

Our algorithm to solve $\LCL$ works as follows.
We start a reachability algorithm for the leader contributor model on the final states that the live computation should visit.
After a modification, the algorithm outputs all interfaces witnessing prefixes to those states.
Let $\reach$ denote the running time of the reachability algorithm.
We show that the obtained set of interfaces will be of size at most $\reach$.
We iterate over the interfaces and pass each to a cycle detection which works over interfaces instead of configurations.
If a cycle was found, a live computation exists.
Let $\cycle$ be the time needed for a single cycle detection.
Then, the running time of the algorithm can be estimated as follows. 
\begin{theorem}
	\label{Theorem:TimeLCL}
	$\LCL$ can be solved in time $\bigO(\reachEmph \cdot \cycleEmph)$.
\end{theorem}

The first step in proving Theorem \ref{Theorem:TimeLCL} is to decompose live computations into prefixes and cycles.
To be precise, we aim for a decomposition where the cycle is saturated in the sense that the initial configuration already contains all contributor states that will be encountered along the cycle.
Knowing these states in advance eases technical arguments when finding cycles in Section \ref{Section:Cycles}.
Formally, a cyclic computation $\tau = c \rightarrow^* c$ is called \emph{saturated} if for each configuration $c'$ in $\tau$, we have $\proj{C}(c') \subseteq \proj{C}(c)$.
We write $c \rightarrow^*_{\sat} c$ for a saturated cycle.
The following lemma yields the desired decomposition.
If not stated otherwise, proofs and details for the current section are provided in Appendix \ref{Section:AppendixInterfaces}.
\begin{lemma}
	\label{Lemma:SplittingLemma}
	There is an infinite initialized computation $\sigma$ with $\Inf(\sigma) \cap F \neq \emptyset$ if and only if there is a finite initialized computation $c^0 \rightarrow^* c \rightarrow^+_{\sat} c$ with $\proj{L}(c) \in F$.
\end{lemma}

We would like to decompose $\LCL$ into finding prefix and cycle.
But we need to ensure that the found computations can be linked at an explicit configuration.
For avoiding the latter, we introduce interfaces.
An \emph{interface} is a triple $I = (S,q,a) \in \Powerset(Q_C) \times Q_L \times \Domain$ consisting of a set of contributor states $S$, a state of the leader $q$, and a memory value $a$.
A configuration $c$ \emph{matches} the interface $I$ if $\proj{C}(c) = S$, $\proj{L}(c) = q$, and $\proj{\Domain}(c) = a$.
We denote this by $I(c)$, interpreting $I$ as a predicate.
The set of interfaces is denoted by $\IF$.
The following lemma shows that the notion allows for decomposing $\LCL$.
We can search for prefixes and cycles separately.
The lemma provides the arguments needed to complete the proof of Theorem \ref{Theorem:TimeLCL}.
\begin{lemma}
	\label{Lemma:GlueComputations}
	Let $I \in \IF$.
	There is a computation $c^0 \rightarrow^* c \rightarrow^+_\sat c$ with $I(c)$ if and only if there are computations $d^0 \rightarrow^* d$ and $f \rightarrow^+_\sat f$ with $I(d) \wedge I(f)$.
\end{lemma}

In the following, we turn to our main contributions.
We present algorithms for reachability and cycle detection and obtain precise values for $\reach$ and $\cycle$.
Further, we modify the reachability algorithms to output interfaces.
Then we invoke Theorem~\ref{Theorem:TimeLCL} to derive algorithms for $\LCL$.
The first problem that we consider is finding prefixes.
\begin{myproblem}
	\problemtitle{Leader Contributor Reachability}
	\problemshort{($\LCR$)}
	\probleminput{A leader contributor system $\lcs = (\Domain, \initmem, P_L, P_C)$ and final states $F \subseteq Q_L$.}
	\problemquestion{Is there an initialized computation $c^0 \rightarrow^* c$ with $\proj{L}(c) \in F$?}
\end{myproblem}

The problem $\LCR$ is $\NP$-complete \cite{Esparza2013}.
Its complexity $\reach$ depends on the parameterization.
There are two standard parameterizations \cite{Chini2018,Meyer2019}:
$\LCR(\sizeL,\sizeD)$ and $\LCR(\sizeC)$.

For the parameterization by $\sizeL$ and $\sizeD$,
we present an algorithm solving $\LCR(\sizeL,\sizeD)$ in time $\leaderdomaintimeEmph$.
The algorithm solves an open problem \cite{Chini2018} by matching the known lower bound:
unless $\ETH$ fails, $\LCR$ cannot be solved in time $2^{o((\sizeL + \sizeD) \cdot \log(\sizeL + \sizeD))}$.
The algorithm and its modification for obtaining interfaces are presented in Section \ref{Section:ParameterizationLeaderDomain}.
\begin{theorem}
	\label{Theorem:LCRLeaderDomain}
	$\LCR(\sizeL,\sizeD)$ can be solved in time $\leaderdomaintimeEmph$.
\end{theorem}

For $\LCR(\sizeC)$, we modify the reachability algorithm from \cite{Chini2018,Meyer2019} so that it outputs interfaces that witness prefixes.
We recall the result on the complexity of the algorithm.
\begin{theorem}[\cite{Chini2018,Meyer2019}]
	\label{Theorem:LCRContributor}
	$\LCR(\sizeC \,)$ can be solved in time $\bigO(\contributortimeReachEmph)$.
\end{theorem}

The second task to solve $\LCL$ is detecting cycles.
We formalize the problem.
It takes an interface and asks for a saturated cycle on a configuration that matches the interface.
\begin{myproblem}
	\problemtitle{Saturated Cycle}
	\problemshort{($\CYC$)}
	\probleminput{A leader contributor system $\lcs = (\Domain, \initmem, P_L, P_C)$ and an interface $I \in \IF$.}
	\problemquestion{Is there a computation $c \rightarrow^+_\sat c$ with $I(c)$?}
\end{myproblem}

We present an algorithm solving $\CYC$ in polynomial time.
Key to the algorithm is a fixed point iteration over certain subgraphs of the contributor.
Details are postponed to Section \ref{Section:Cycles}.
\begin{theorem}
	\label{Theorem:CyclePolyTime}
	$\CYC$ can be solved in time $\bigO(\fixpointtimeEmph)$.
\end{theorem}

The theorem shows that $\cycle$ is polynomial.
Hence, by Theorem \ref{Theorem:TimeLCL}, we obtain that $\LCL$ can be solved in time $\bigOS(\reach)$.
This means that liveness verification and safety verification in the leader contributor model only differ by a polynomial factor.
Taking the precise values for $\reach$ into account, Theorem \ref{Theorem:TimeLCL} yields the following.
\begin{corollary}
	\label{Corollary:LCLLEaderDomainTime}
	$\LCL(\sizeL,\sizeD)$ can be solved in time $\leaderdomaintimeEmph$.
\end{corollary}
\begin{corollary}
	\label{Corollary:LCLContributorTime}
	$\LCL(\sizeC\,)$ can be solved in time $\bigO(\contributortimeEmph)$.
\end{corollary}

For the latter result, we are actually more precise in determining the time complexity than stated in Theorem \ref{Theorem:TimeLCL}.
Both obtained algorithms are optimal.
They match the corresponding lower bounds for $\LCL$ that carry over from reachability \cite{Chini2018}.
Unless $\ETH$ fails, $\LCL$ cannot neither be solved in time $2^{o((\sizeL + \sizeD) \cdot \log(\sizeL + \sizeD))}$ nor in time $2^{o(\sizeC)}$.

\section{Reachability Parameterized by Leader and Domain}
\label{Section:ParameterizationLeaderDomain}

We present the algorithm for $\LCR(\sizeL,\sizeD)$.
It runs in time $\leaderdomaintimeEmph$ and therefore proves Theorem \ref{Theorem:LCRLeaderDomain}.
Moreover, with the results from Section \ref{Section:Interfaces} and \ref{Section:Cycles}, the algorithm can be utilized for solving $\LCL$ in time $\leaderdomaintimeEmph$.
Like in \cite{Chini2018}, the algorithm relies on a notion of witnesses.
These are sketches of computations.
A witness is valid if there is an actual computation following the sketch.
Validity can be checked in polynomial time.

The algorithm from \cite{Chini2018} iterates over all witnesses and tests validity for each.
Hence, the time complexity of the algorithm is proportional to $(\sizeL \sizeD)^{\bigO(\sizeL \sizeD)}$, the number of considered witnesses.
Key to our new algorithm is the fact that we can restrict to so-called short witnesses.
These are sketches of loop-free computations.
We show that validity of witnesses can be checked inductively from validity of short witnesses.
We exploit the inductivity by a dynamic programming.
It runs in time proportional to $\leaderdomaintimeEmph$, the number of short witnesses.
This yields the desired complexity as stated in Theorem \ref{Theorem:LCRLeaderDomain}.

\subsection{Witnesses and Validity}
\label{Section:Witnesses}

We introduce witnesses and recall the notion of validity.
Afterwards, we elaborate on the main idea of our new algorithm:
restricting to short witnesses for checking validity.

Intuitively, a witness is a compact way to represent computations of a leader contributor system.
From a computation, a witness only stores the actions of the leader and the positions where memory symbols were written by a contributor for the first time.
We call these positions \emph{first writes}.
From such a position on, we can assume an unbounded supply of the corresponding memory symbol.
There is always a copy of a contributor waiting to provide it.

Formally, a \emph{witness} is a triple $x = (w,q,\sigma)$.
The word $w = (q_1, a_1)  (q_2 , a_2) \ldots (q_n, a_n)$ represents the run of the leader.
It is a sequence from $(Q_L \times (\Domain \uplus \setcon{\bot}))^*$, containing leader states potentially combined with a memory value.
The state $q \in Q_L$ is the target of the leader run.
First-write positions are specified by $\sigma: [1..k] \rightarrow [1..n]$, a monotonically increasing map where $k \leq \sizeD$.
The number of first-write positions $k$ is called the \emph{order} of $x$.
We denote it by $\ord(x) = k$.
Moreover, we use $\Wit$ for the set of all witnesses.
A witness $x = (w,q,\sigma) \in \Wit$ is called \emph{initialized} if $w$ begins in the initial state $q^0_L$ of the leader automaton.

If a witness corresponds to an actual computation, we call it valid.
This means, the witness encodes a proper run of the leader and moreover, the first writes along the run can be provided by the contributors.
Since the definition of witnesses only specifies first-write positions but not values, we need the notion of first-write sequences.
The latter will allow for the definition of validity.
 
A \emph{first-write sequence} is a sequence of data values $\beta \in \Domain^{\leq \sizeD}$ that are all different.
Formally, $\beta_i\neq \beta_j$ for $i\neq j$. 
We use $\FW$ to denote the set of all those sequences.
Given a witness $x = (w,q,\sigma)$, we define its validity with respect to a first-write sequence $\beta$ of length $\ord(x)$.
For being valid, $x$ has to be \emph{leader valid along} $\beta$ and \emph{contributor valid along} $\beta$.
We make both notions more precise.
Details regarding this section including formal definitions are available in Appendix \ref{Section:AppendixParameterizationLeaderDomain}.
 
\subparagraph*{Leader Validity.}
The witness is \emph{leader valid along} $\beta$ if $w$ encodes a run of the leader that reaches state $q$.
Reading during the run is restricted to symbols from $\beta$:
the $\ell$-th symbol $\beta_\ell$ is available for reading once the run arrives at position $\sigma(\ell)$.
Formally, the encoding depends on the memory values $a_i$.
If $a_i \neq \bot$, the leader has a transition $q_i \XXrightarrow{\wt{a_i}}{L} q_{i+1}$.
If $a_i = \bot$, the leader either has an $\varepsilon$-transition or reads a symbol available at position $i$, from the set $S_\beta(i) = \Setcon{\beta_{\ell}}{\sigma(\ell) \leq i}$.
We use $\LValid_\beta(x)$ to indicate that $x$ is leader valid along $\beta$.
 
\subparagraph*{Contributor Validity.}
The witness is \emph{contributor valid along} $\beta$ if the contributors can provide the first writes for $w$ in the order indicated by~$\sigma$.
Let us focus on the $i$-th first write~$\beta_i$.
Providing $\beta_i$ is a question of reachability of the set $Q_i = \Setcon{p}{\exists p' : p \XXrightarrow{\wt{\beta_i}}{C} p' }$ in the contributor automaton.
More precise, we need a contributor that reaches $Q_i$ while reading only symbols available along $w$.
This means that reading is restricted to earlier first writes and symbols written by the leader during $w$ up to position $\sigma(i)$.

Let $\Expr(x, \beta_1\dots \beta_{i-1})$ be the language of available reads.
We say that $x$ is \emph{valid for the $i$-th first write of} $\beta$ if $Q_i$ is reachable by a contributor while reading is restricted to $\Expr(x, \beta_1\dots \beta_{i-1})$.
We use  $\CValid^i_\beta(x)$ to indicate this validity.
If $x$ is valid for all first writes, it is \emph{contributor valid along} $\beta$.
Formally, $\CValid_\beta(x) = \bigwedge_{i \in [1..\ord(x)]} \CValid^i_\beta(x)$.

With leader and contributor validity in place, we can define $x$ to be \emph{valid along} $\beta$ if $\LValid_\beta(x) \wedge \CValid_\beta(x)$.
Again, we use predicate notation.
We write $\Valid_\beta(x)$ if $x$ is valid along $\beta$.
Validity of a witness along a first-write sequence can be checked in polynomial time.
\begin{lemma}
	\label{Lemma:ValidityPolyTime}
	Let $x \in \Wit$ and $\beta \in \FW$.
	$\Valid_\beta(x)$ can be evaluated in polynomial time.
\end{lemma}

The algorithm from \cite{Chini2018} iterates over witnesses and invokes Lemma \ref{Lemma:ValidityPolyTime} to check validity.
The following lemma proves the correctness:
validity indicates the existence of a computation.
\begin{lemma}
	\label{Lemma:Witnesses}
	Let $q\in Q_L$. 
	There is an initialized computation $c^0 \rightarrow^* c$ with $\proj{L}(c) = q$ if and only if there is an initialized $x = (w, q, \sigma)\in \Wit$ and a $\beta \in \mathit{FW}$ so that $\Valid_\beta(x)$.
\end{lemma}

For obtaining a tractable algorithm, we would like to restrict to short witnesses when checking validity.
These are witnesses encoding a loop-free run of the leader.
The following two observations are crucial to our development.

Leader validity can be checked inductively on short witnesses.
A witness $x$ can be written as a product $x = x_1 \times x_2 \times \dots \times x_{k+1}$ of smaller witnesses.
Each $x_i$ encodes that part of the leader run of $x$ happening between two first-write positions $\sigma(i-1)$ and $\sigma(i)$.
The \emph{witness concatenation} $\times$ appends these runs.
Each $x_i$ can assumed to be a short witness.
There is no need for recording loops of the leader between first writes.
We can cut them out.

Assume $y = x_1 \times \dots \times x_i$ encodes a proper run $\rho$ of the leader that reads from the available first writes $\beta_1, \dots, \beta_{i-1}$.
Formally, $\LValid_{\beta_1 \dots \beta_{i-1}}(y)$.
Then, leader validity of $y \times x_{i+1}$ along $\beta_1 \dots \beta_i$ mainly depends on the newly added witness $x_{i+1}$.
The reason is that we prolong~$\rho$, a run of the leader that was already verified.
All that we have to remember from $\rho$ is where it ends.
This means that we can shrink $y$ to a short witness.
We consecutively cut out loops from the leader, denoted by $\Shrink^*$, until we obtain a loop free witness.
Formally, if $\LValid_{\beta_1 \dots \beta_{i-1}}(y)$ holds true, we have the equality
\begin{align*}
	\LValid_{\beta_1 \dots \beta_i}(y \times x_{i+1}) = \LValid_{\beta_1 \dots \beta_i}(\Shrink^*(y) \times x_{i+1}).
\end{align*}
Hence, checking leader validity can be restricted to (concatenations of) short witnesses.

Like leader validity, we can restrict contributor validity to short witnesses.
The main reason is that testing validity for the $i$-th first write only requires limited knowledge about earlier first writes.
As long as we guarantee that earlier first writes can be provided along a run of the leader, we do not have to keep track of their precise positions anymore.
This means that we can shrink the run when testing validity for the $i$-th first write.

Assume that $y = x_1 \times \dots \times x_i$ is known to be contributor valid.
Formally, $\CValid_{\beta_1 \dots \beta_{i-1}}(y)$ is true.
Note that the first writes considered in $y$ are $\beta_1, \dots, \beta_{i-1}$.
We want to check contributor validity of $y \times x_{i+1}$.
Since there is only one new first write that we add, namely~$\beta_i$, we have to evaluate $\CValid^i_{\beta_1 \dots \beta_i}(y \times x_{i+1})$.
Satisfying contributor validity means that $\beta_i$ can be provided along $y \times x_{i+1}$ assuming that $\beta_1, \dots, \beta_{i-1}$ were already provided.
In fact, it is not important where these earlier first writes appeared exactly.
We just need the fact that after $y$, they can assumed to be there.
This allows for shrinking $y$ and forgetting about the precise positions of the earlier first writes.
Formally, if $\CValid_{\beta_1 \dots \beta_{i-1}}(y)$, we have
\begin{align*}
	\CValid^i_{\beta_1 \dots \beta_i}(y \times x_{i+1}) = \CValid^i_{\beta_1 \dots \beta_i}(\Shrink^*(y) \times x_{i+1}).
\end{align*}

In the next section, we turn the above observations into a recursive definition of validity for short witnesses.
The recursion only involves short witnesses of lower order.
Since the number of these is bounded by $\leaderdomaintimeEmph$, we can employ a dynamic programming that checks validity of short witnesses in time proportional to their number.

\subsection{Algorithm and Correctness}
\label{Section:AlgorithmAndCorrectness}

Before we can formulate the recursion, we need to introduce short witnesses and a concatenation operator on the same.
A \emph{short witness} is a witness $z = (w,q,\sigma) \in \Wit$ where the leader states in $w = (q_1, a_1)\dots(q_n, a_n)$ are all distinct.
We use $\SWit$ to denote the set of all short witnesses.
Moreover, let $\Ord(k)$ denote the set of those short witnesses that are of order $k$.

Let $x = (w,q,\sigma) \in \Ord(i)$ and $y = (w',q',\sigma') \in \Ord(j)$ be two short witnesses.
Assume that the first state in $w'$ is $q$, meaning that $y$ starts with the target state of $x$.
Then, the \emph{short concatenation} of $x$ and $y$ is defined to be the short witness $x \otimes y = \Shrink^*(x \times y) \in \Ord(i+j)$.

The price to pay for the smaller number of short witnesses is a more expensive check for validity.
Rather than checking validity once for each short witness, we build them up by a recursion along the order, and check validity for each composition.
Let $z$ be a short witness.
If $\ord(z) = 0$, there are no first-write positions.
Only leader validity is important:
\begin{align*}
	\SValid_\varepsilon(z) =  \LValid_\varepsilon(z).
\end{align*}
For a short witness $z$ of order $k+1$, we define validity along $\beta = \beta_1 \dots \beta_{k+1} \in \FW$ by
\begin{align*}
	\SValid_\beta(z)  =  \bigvee_{\substack{x \in \Ord(k)\\  y \in \Ord(1)}}
	[z = x \otimes y] \wedge
	\LValid_\beta(x \times y ) \wedge
	\CValid^{k+1}_\beta(x \times y) \wedge
	\SValid_{\beta'}(x).
\end{align*}
Here $\beta' = \beta_1 \dots \beta_k$ is the prefix of $\beta$ where the last element is omitted.

The idea behind the recursion is to cut off the last first write $\beta_{k+1}$, check its validity, and recurse on the remaining part.
To this end, $z$ is decomposed into two short witnesses $x \in \Ord(k)$ and $y \in \Ord(1)$.
Intuitively, $x$ is the compression of a larger witness that is already known to be valid and $y$ is the short witness responsible for the last first write.
By our considerations above, we already know that it suffices to check validity for $\beta_{k+1}$ with $x$ instead of its expanded form.
These are the evaluations $\LValid_\beta(x \times y)$ and $\CValid^{k+1}_\beta(x \times y)$.
To guarantee validity along $\beta'$, we recurse on $\SValid_{\beta'}(x)$.

The following lemma shows the correctness of the recursion.
Using Lemma~\ref{Lemma:Witnesses}, we can work with short witnesses to discover computations in the given leader contributor system.
\begin{lemma}
	\label{Lemma:RestrictionToShortWitnesses}
	Let $q \in Q_L$ and $\beta \in \FW$.
	There is an $x = (w,q,\sigma) \in \Wit$ with $\Valid_\beta(x)$ if and only if there is an $z = (w',q,\sigma') \in \SWit$ with $\SValid_\beta(z)$.
	In this case, $\init(x) = \init(z) $.
\end{lemma}
Note that in the lemma, $\init(x)$ refers to the first state of $w$.
Similarly for $z$.

It remains to give the algorithm.
For each first-write sequence $\beta$ and each short witness~$z$, we compute $\SValid_\beta(z)$ by a dynamic programming.
To this end, we maintain a table indexed by first-write sequences and short witnesses.
An entry for $\beta \in \FW$ and $z \in \SWit$ is computed as follows.
Let $\abs{\beta} = \ord(z) = k$.
We iterate over all short witnesses $x \in \Ord(k-1), y \in \Ord(1)$ and check whether $z = x \otimes y$ holds.
If so, we compute $\LValid_\beta(x \times y) \wedge \CValid^k_\beta(x \times y)$ and look up the value of $\SValid_{\beta'}(x)$ in the table.
Details on the precise complexity are presented in Appendix \ref{Section:AppendixParameterizationLeaderDomain}.
\begin{proposition}
	\label{Proposition:ValidWitnesses}
	The set of all valid short witnesses can be computed in time $\leaderdomaintimeEmph$.
\end{proposition}
It is left to explain how interfaces can be obtained from the algorithm.
From a valid short witness, target state and last memory value can be read off.
Contributor states can be obtained by synchronizing the contributor along the witness.
This takes polynomial time.
Details can be found in Appendix \ref{Section:AppendixParameterizationLeaderDomain}.

\section{Finding Cycles in Polynomial Time}
\label{Section:Cycles}

We give an efficient algorithm solving $\CYC$ in time $\bigO(\fixpointtime)$.
This proves Theorem~\ref{Theorem:CyclePolyTime}.
The algorithm relies on a characterization of cycles in terms of stable SCC decompositions.
These are decompositions of the contributor automaton into strongly connected subgraphs that are stable in the sense that they write exactly the symbols they intend to read.
With a fixed point iteration, we show how to find stable SCC decompositions in the mentioned time.

Our algorithm is technically simple.
It relies on a fixed point iteration calling Tarjan's algorithm \cite{Tarjan1972} to obtain SCC decompositions.
Hence, the algorithm is easy to implement and shows that stable SCC decompositions are the ideal structure for detecting cycles.
Moreover, we can modify the algorithm to detect cycles where the leader necessarily makes a move.

We also discovered that cycles can be detected by a non-trivial polynomial-time reduction to the problem of finding cycles in dynamic graphs.
Although the latter can be solved in polynomial time \cite{Kosaraju1988}, the obtained algorithm for $\CYC$ does not admit an efficient polynomial-time complexity.
The reason is that the algorithm in \cite{Kosaraju1988} repeatedly solves linear programs that grow large due to the reduction. 
Compared to this method, our algorithm is more efficient and technically simpler due to being tailored to the actual problem.

\subsection{From Saturated Cycles to Stable SCC decompositions}
\label{Section:SatCycleToSCC}

We characterize cycles in terms of stable SCC decompositions.
These are decompositions of the contributor automaton that can provide themselves with all the symbols that a cycle along this structure may read.
For the definition, we generalize properties of a fixed cycle to the fact that a saturated cycle exists.
We link the latter with an alphabet $\Gamma$, a variable for the set of reads in a saturated cycle.
Then we define stable SCC decompositions depending on $\Gamma$.
Hence, the search for a cycle amounts to finding a $\Gamma$ with a stable SCC decomposition.

Throughout the section, we fix an interface $I = (S,q,a)$ and a saturated cycle $\tau = c \rightarrow^+_\sat c$ with $I(c)$.
We assume that the set $\Writes(\tau) = \Setcon{b \in \Domain}{ d \xxrightarrow{\wt{b}} d' \in \tau }$ is non-empty, $\tau$ contains at least one write.
If $\tau$ contains only reads, then either a contributor or the leader run in an $\rd{a}$-loop, a cycle which is easy to detect.
We generalize two properties of $\tau$.

\subparagraph*{Property 1: Strongly connectedness.}
Considering the saturated cycle $\tau$, we can observe how the current state of a particular contributor $P$ changes over time.
Assume $P$ starts in a state $p$ and visits a state $p'$ during $\tau$.
Since it runs along the cycle, the contributor will eventually move from $p'$ back to $p$ again.
This means that in the contributor automaton, there is a path from $p$ to $p'$ and vice versa.
Phrased differently, $p$ and $p'$ are strongly connected.

To make this notion more precise, we define a subgraph of the contributor automaton.
Intuitively, it is the restriction of $P_C$ to the states and transitions visited along $\tau$.
Rather than defining it for a single computation $\tau$, we generalize to a set of \emph{enabled reads} $\Gamma \subseteq \Domain$.
The directed graph $\restrictionGraph{S}{\Gamma} = (S,E(\Gamma))$ has as vertices the contributor states $S$ and as edges the set $E(\Gamma)$.
The latter are transitions of $P_C$ between states in $S$ that are either reads enabled by $\Gamma$ or writes of arbitrary symbols.
Formally, we have
\begin{align*}
	(p,p') \in E(\Gamma) \text{~if~} p \XXrightarrow{\rd{b}}{C} p' \text{~with~} b \in \Gamma \text{~or~} p \XXrightarrow{\wt{b}}{C} p' \text{~with~} b \in \Domain.
\end{align*}

For the cycle $\tau = c \rightarrow^+_\sat c$, the induced graph is $\restrictionGraph{S}{\Gamma}$ where $\Gamma = \Writes(\tau)$.
With the graph in place, we can define our notion of strongly connected states.
\begin{definition}
	\label{Definition:StronglyConnected}
	Let $p,p' \in S$ be two states and $\Gamma \subseteq \Domain$.
	We say that $p$ and $p'$ are \emph{strongly $\Gamma$-connected} if $p$ and $p'$ are strongly connected in the graph $\restrictionGraph{S}{\Gamma}$.
\end{definition}
Like the classical notion, the above definition generalizes to sets.
We say that a set $V \subseteq S$ is \emph{strongly $\Gamma$-connected} if each two states in $V$ are strongly $\Gamma$-connected.

The saturated cycle $\tau$ runs along the SCC decomposition of its induced graph $\restrictionGraph{S}{\Gamma}$.
Following a particular contributor $P$ in $\tau$, we collect the visited states in a set $S_P \subseteq S$.
Then, $S_P$ is strongly $\Gamma$-connected and thus contained in an inclusion maximal strongly connected set, an SCC of $\restrictionGraph{S}{\Gamma}$.
Hence, the contributors in $\tau$ stay within SCCs of the graph.
We associate with $\tau$ the SCC decomposition.
Again, we generalize to a given alphabet.

Let $\Gamma \subseteq \Domain$ and $V \subseteq S$ strongly $\Gamma$-connected.
We call $V$ a \emph{strongly $\Gamma$-connected component} ($\Gamma$-SCC) if it is inclusion maximal.
The latter means that for each $V \subseteq V'$ with $V'$ strongly $\Gamma$-connected, we already have $V = V'$.
We consider the unique partition of $S$ into $\Gamma$-SCCs.
Note that by a partition, we mean a collection $(S_1, \dots, S_\ell)$ of pairwise disjoint subsets of $S$ such that $S = \bigcup_{i\in[1..\ell]} S_i$.
The order of a partition is not important for our purpose.
\begin{definition}
	\label{Definition:SCCdecomposition}
	The partition of $S$ into $\Gamma$-SCCs is called \emph{$\Gamma$-SCC decomposition} of $S$.
\end{definition}
We denote the $\Gamma$-SCC decomposition by $\decomp{S}{\Gamma}$.
It consists of the vertices of the SCC decomposition of $\restrictionGraph{S}{\Gamma}$.
Hence, we can obtain it from an application of Tarjan's algorithm~\cite{Tarjan1972}, a fact that becomes important when computing $\decomp{S}{\Gamma}$ in Section \ref{Section:ComputingSCC}.

\subparagraph*{Property 2: Stability.}
Let $\decomp{S}{\Gamma} = (S_1, \dots, S_\ell)$ be the $\Gamma$-SCC decomposition associated with the saturated cycle $\tau$.
The writes in $\tau$ can be linked with the $S_i$.
If a write occurs between states $p,p' \in S_i$, we associate it with the set $S_i$.
The writes of the leader all occur on a cyclic computation $q \rightarrow^*_L q$.
The point of assigning writes to sets is the following.
Writes that belong to a set can occur on a cycle through a set of the decomposition.

We generalize from $\tau$ to a given alphabet $\Gamma\subseteq \Domain$. 
Let $\decomp{S}{\Gamma} = (S_1, \dots, S_\ell)$ be the $\Gamma$-SCC decomposition of $S$. 
The \emph{writes of the decomposition} is the set of all symbols that occur as writes either between the states of $S_i$ or in a cycle $q \rightarrow^*_L q$ on the leader while preserving the memory content $a$.
Formally, we define the writes to be the union $\Writes(S_1, \dots, S_\ell) = \Writes_C(S_1, \dots, S_\ell) \cup \Writes_L(S_1, \dots, S_\ell)$ where
\begin{align*}
	\Writes_C(S_1, \dots, S_\ell) &= \Setcon{b}{p \XXrightarrow{\wt{b}}{C} p' \text{ with } p,p' \in S_i} \text{ and } \\
	\Writes_L(S_1, \dots, S_\ell) &= \Setcon{b}{\exists u,v : (q,a) \XXrightarrow{u.!b.v}{L'} (q,a)}.
\end{align*}
Here, $\rightarrow_{L'}$ denotes the transition relation of the automaton $P_{L'}$, a restriction of the leader~$P_L$ to reads within $\Writes_C(S_1, \dots, S_\ell)$. 
The automaton also keeps track of the memory content.  
We define $P_{L'} = (\Ops{\Domain}, Q_L \times \Domain, (q^0_L,\initmem), \delta_{L'})$ with the transitions
\begin{align*}
	(s,b) &\XXrightarrow{\wt{b'}}{L'} (s',b') &\text{ if }& s \XXrightarrow{\wt{b'}}{L} s', \\
	(s,b) &\XXrightarrow{\rd{b}}{L'} (s',b) &\text{ if }& s \XXrightarrow{\rd{b}}{L} s' \text{ and } b \in \Writes_C(S_1, \dots, S_\ell), \\
	(s,b) &\XXrightarrow{\varepsilon}{L'} (s,b') &\text{ if }& b' \in \Writes_C(S_1, \dots, S_\ell).
\end{align*}
The last transitions change the memory content due to a write of a contributor.

The following lemma states that writes behave monotonically.
This fact will become important in Section \ref{Section:ComputingSCC}.
We provide a proof in Appendix \ref{Section:AppendixCycles}.
\begin{lemma}
	\label{Lemma:Monotonicity}
	Let $\Gamma \subseteq \Gamma' \subseteq \Domain$.
	We have $\Writes(\decomp{S}{\Gamma}) \subseteq \Writes(\decomp{S}{\Gamma'})$.
\end{lemma}

During the cycle $\tau$, reads are always preceded by corresponding writes.
Hence, the writes of the $\Gamma$-SCC decomposition, where $\Gamma = \Writes(\tau)$, provide all symbols needed for reading.
In fact, we have $\Writes(\decomp{S}{\Gamma}) \supseteq \Gamma$.
The following definition generalizes this property.
\begin{definition}
	\label{Definition:StableSCCDecomp}
	Let $\Gamma \subseteq \Domain$. 
	The $\Gamma$-SCC decomposition $\decomp{S}{\Gamma}$ of $S$ is called \emph{stable} if it provides $\Gamma$ as its writes, meaning $\Writes(\decomp{S}{\Gamma}) = \Gamma$.
\end{definition}
Note that the definition asks for equality instead of inclusion.
The reason is that we can express stability as a fixed point of a suitable operator. 
This will be essential in Section \ref{Section:ComputingSCC}.

\subparagraph{Characterization.}
The following proposition characterizes the existence of saturated cycles via stable SCC decompositions.
It is a major step towards the polynomial-time algorithm.
\begin{proposition}
	\label{Proposition:CharacterizationStableSCCDecomp}
	There is a saturated cycle $\tau = c \rightarrow^+_\sat c$ with $I(c)$ if and only if there exists a non-empty subset $\Gamma \subseteq \Domain$ such that $\decomp{S}{\Gamma}$ is stable.
\end{proposition}
\begin{proof}
	Assume the existence of a saturated cycle $\tau$.
	Our candidate set is $\Gamma = \Writes(\tau)$.
	We already argued above that $\Writes(\decomp{S}{\Gamma}) \supseteq \Gamma$.
	If equality holds, $\decomp{S}{\Gamma}$ is stable and $\Gamma$ is the set we are looking for.
	Otherwise, we have $\Writes(\decomp{S}{\Gamma}) \supsetneq \Gamma$.
	
	In the latter case, we consider $\Gamma' = \Writes(\decomp{S}{\Gamma})$ instead of $\Gamma$.
	Since $\Gamma' \supseteq \Gamma$, we can apply Lemma \ref{Lemma:Monotonicity} and obtain that $\Writes(\decomp{S}{\Gamma'})$ contains $\Gamma'$.
	
	Iterating this process yields a sequence of sets $(\Gamma_i)_i$ that is strictly increasing, $\Gamma_i \subsetneq \Gamma_{i+1}$, and that satisfies $\Writes(\decomp{S}{\Gamma_i}) \supseteq \Gamma_i$.
	The sequence is finite since $\Gamma_i \subseteq \Domain$ for all $i$.
	Hence, there is a last set $\Gamma_d$ which necessarily fulfills $\Writes(\decomp{S}{\Gamma_d}) = \Gamma_d$.
	
	For the other direction, we need to construct a saturated cycle from a set $\Gamma$ with stable SCC decomposition.
	Idea and formal proof are given in Appendix \ref{Section:AppendixCycles}.
\end{proof}

\subsection{Computing Stable SCC decompositions}
\label{Section:ComputingSCC}

The search for a saturated cycle reduces to finding an alphabet $\Gamma$ with a stable SCC decomposition.
Following the definition of stability, we can express $\Gamma$ as a fixed point that can be computed by a Kleene iteration \cite{Winskel1993} in polynomial time.
We define the suitable operator.
It acts on the powerset lattice $\Powerset(\Domain)$ and for a given set $X$, it computes the writes of the $X$-SCC decomposition.
Formally, it is defined by 
\begin{align*}
	\WritesSCC(X) = \Writes(\decomp{S}{X}).
\end{align*}
The operator is monotone and can be evaluated in polynomial time.
\begin{lemma}
	\label{Lemma:OperatorProperties}
	For $X \subseteq X'$ subsets of $\Domain$, we have $\WritesSCC(X) \subseteq \WritesSCC(X')$.
	Moreover, $\WritesSCC(X)$ can be computed in time $\bigO(\evaltimeEmph)$.
\end{lemma}

Monotonicity follows from Lemma \ref{Lemma:Monotonicity}.
For the evaluation, let $X$ be given.
We apply Tarjan's algorithm on $\restrictionGraph{S}{X}$ to compute the $X$-SCC decomposition $\decomp{S}{X}$.
This takes linear time.
It is left to compute the writes $\Writes(\decomp{S}{X})$.
For details on the computation and the precise complexity we refer to Appendix \ref{Section:AppendixCycles}.

The following lemma states that the non-trivial fixed points of the operator $\WritesSCC$ are precisely the sets with a stable SCC decomposition.
Hence, searching for a cycle reduces to searching for a fixed point.
\begin{lemma}
	\label{Lemma:FixedPoint}
	For $\Gamma \neq \emptyset$ we have, $\Gamma = \WritesSCC(\Gamma)$ if and only if $\decomp{S}{\Gamma}$ is stable.
\end{lemma}

Correctness immediately follows from the definition of stability.
For finding a suitable set $\Gamma$, we employ a Kleene iteration to compute the greatest fixed point of $\WritesSCC$.
It starts from $\Gamma = \Domain$, the top element of the lattice.
At each step, it evaluates $\WritesSCC(\Gamma)$ by invoking Lemma~\ref{Lemma:OperatorProperties}.
This takes time $\bigO(\evaltime)$.
Termination is after at most~$\sizeD$ steps since at least one element is removed from the set $\Gamma$ each iteration.
Hence, the time to compute the greatest fixed point of $\WritesSCC$ is $\bigO(\fixpointtime)$.

\section{Conclusion}
\label{Section:Conclusion}

We studied the fine-grained complexity of $\LCL$, the liveness verification problem for leader contributor systems.
To this end, we first decomposed $\LCL$ into the reachability problem $\LCR$ and the cycle detection $\CYC$.
We focused on the complexity of $\LCR$.
While an optimal $\bigOS(2^\sizeC)$-time algorithm for $\LCR(\sizeC)$ was already known, we presented an algorithm solving $\LCR(\sizeL,\sizeD)$ in time $\leaderdomaintimeEmph$.
The algorithm is optimal in the fine-grained sense and therefore solves an open problem.
It is a dynamic programming based on a notion of valid short witnesses.
Moreover, we showed how to modify both algorithms for $\LCR$ so that they are compatible with a cycle detection and can be used in algorithms solving $\LCL$.

Further, we determined the complexity of $\CYC$.
We presented an efficient fixed point iteration running in time $\bigO(\fixpointtime)$.
It is based on a notion of stable SCC decompositions and invokes Tarjan's algorithm to find them.
The result shows that $\LCL$ and $\LCR$ admit the same fine-grained complexity.

\subparagraph*{Acknowledgments.}
We thank Arnaud Sangnier for helpful discussions.

\bibliographystyle{plain}
\bibliography{content/cite}

\newpage
\appendix
\section*{Appendix}
\label{Section:Appendix}

\section{Proofs of Section \ref{Section:Interfaces}}
\label{Section:AppendixInterfaces}

We provide proofs and details for Section \ref{Section:Interfaces}.

\subsection*{Proof of Lemma \ref{Lemma:SplittingLemma}}
Given a computation $c^0 \rightarrow^* c \rightarrow^+_{\sat} c$ such that $\proj{L}(c) \in F$, we can iterate the cyclic part to obtain a computation that visits $F$ infinitely often.
For the other direction, let $\sigma$ be an infinite initialized computation with $\Inf(\sigma) \cap F \neq \emptyset$.
Then, $\sigma$ visits infinitely many configurations involving a state from $F$.
These constitute an infinite sequence over the finite set $\Conf^t$. 
Hence, there is a repeating configuration $c$ and we get \mbox{$c^0 \rightarrow^* c \rightarrow^+ c$ with $\proj{C}(c) \in F$.}

It is left to show that we can assume a saturated cycle.
We use an idea going back to the \emph{copycat lemma} \cite{Esparza2013}.
Suppose $c \rightarrow^+ c$ is not saturated. 
Then there is a state $p \in Q_C$ which does not occur in $c$ but is encountered in a configuration $c'$ on the cycle. 
Let $P$ denote the contributor that visits $p$ in $c'$.
We add a new contributor $P^\cop$ to the computation that mimics the behavior of $P$.
Each time $P$ takes a transition, $P^\cop$ copycats it immediately.
Once $P^\cop$ reaches $p$, it does not move any further and stays in $p$.
We apply the procedure for each new state occurring in the cycle.
After having iterated through the cycle, we have collected all these states and there is a contributor staying in each of them.
Now we can run the cycle without discovering new states. 
This yields $d^0 \rightarrow^* d \rightarrow^+_\sat d$ with $\proj{L}(d) \in F$, as required.

\subsection*{Proof of Lemma \ref{Lemma:GlueComputations}}
Before we give the proof, we introduce a notion for counting contributor states in a configuration.
Let $c = (q,a,\pc) \in \Conf^t$ with $t \in \Naturals$ be any configuration and $p \in Q_C$ a contributor state.
The \emph{cardinality} $\card_p(c)$ denotes the number of contributors in configuration $c$ the current state of which is $p$.
Formally, we define
\begin{align*}
	\card_p(c) = \abs{\Setcon{i \in [1..t]}{\pc(i) = p}}.
\end{align*}
We proceed with the proof of Lemma \ref{Lemma:GlueComputations}.
\begin{proof}
	If we are given a computation of the form $c^0 \rightarrow^* c \rightarrow^+_\sat c$ with $I(c)$, we split it into the prefix $c^0 \rightarrow^* c$ and the cycle $c \rightarrow^+_\sat c$.
	The interface $I$ is clearly matched.
	
	For the other direction, let computations $d^0 \rightarrow^* d$ and $f \rightarrow^+_\sat f$ with $I(d) \wedge I(f)$ be given.
	We construct a composed computation $c^0 \rightarrow^* c \rightarrow^+_\sat c$ with $I(c)$ as desired.
	
	Let $c$ be a configuration that contains for each state $p$ the maximal amount of contributors of $d$ and $f$ that are currently in $p$.
	Memory and leader state are identical to $d$ and $f$.
	Formally we have, $\card_p(c) = \max(\card_p(d),\card_p(f))$ for each state $p \in Q_C$.
	Moreover, $\proj{L}(c) = \proj{L}(d)$ and $\proj{\Domain}(c) = \proj{\Domain}(d)$.
	This implies $I(c)$.
	
	In the following, we show that a live computation involving $c$ can be obtained by the given prefix and cycle.
	By the copycat lemma, we can enrich the computation $d^0 \rightarrow^* d$ by contributors such that we get $c^0 \rightarrow^* c$.
	In fact, if we have that \mbox{$\max(\card_p(d),\card_p(f)) = \card_p(d)$}, we do not have to add contributors for state $p$.
	If $\max(\card_p(d),\card_p(f)) > \card_p(d)$, we add contributors for the difference $t = \max(\card_p(d),\card_p(f)) - \card_p(d)$.
	Let $P$ be any contributor in $d$ currently in state $p$.
	Then, we add $t$ copies $P^\cop_1, \dots, P^\cop_t$ of $P$ to $d$.
	Since the behavior of the leader and the memory do not change, we get the prefix $c^0 \rightarrow^* c$.
	
	The cycle $f \rightarrow^+_\sat f$ can be simulated on the larger configuration $c$.
	Intuitively, the contributors that do not participate in the cycle, can be ignored.
	Hence, we obtain the desired cycle $c \rightarrow^+_\sat c$.
	Note that it is saturated.
	This completes the proof.
\end{proof}

\subsection*{Proof of Theorem \ref{Theorem:TimeLCL}}
We assume that we have already modified the reachability algorithm so that it computes all interfaces that witness a prefix computation.
Moreover, this is possible in time $\reach$ and there are at most $\reach$ such interfaces.
We prove this assumption to be correct when considering corresponding reachability algorithms.

We first show the correctness of the algorithm.
Each interface $I$ that we iterate over witnesses the existence of a prefix computation $d^0 \rightarrow^* d$ with $I(d)$.
If $I$ is a positive instance of the cycle detection, we get a saturated cycle $f \rightarrow^+_\sat f$ which satisfies $I(f)$.
By Lemma \ref{Lemma:GlueComputations}, we then get a computation of the form $c^0 \rightarrow^* c \rightarrow^+_\sat c$ with $I(c)$.
Hence, by Lemma \ref{Lemma:SplittingLemma}, we obtain a live computation.

On the other hand, let a live computation be given.
By Lemma \ref{Lemma:SplittingLemma} we can assume it to be of the shape $c^0 \rightarrow^* c \rightarrow^+_\sat c$.
We let $I = (S,q,a)$ be the interface induced by $c$.
Formally, $S = \proj{C}(c)$, $q = \proj{L}(c)$, and $a = \proj{\Domain}(c)$.
Since $I$ witnesses the prefix $c^0 \rightarrow^* c$, the algorithm iterates over $I$ and passes it to the cycle detection.
Since the cycle $c \rightarrow^+_\sat c$ satisfies $I(c)$, the cycle detection accepts interface $I$ and the algorithm returns \emph{yes}.

The complexity of the algorithm can be estimated as follows.
We compute all interfaces witnessing prefix computations by a call to the modified reachability algorithm.
This takes time $\reach$.
Since we assume that there are at most $\reach$ such interfaces, iterating over them and passing each to the cycle detection takes time \mbox{$\reach \cdot \cycle$.}
Summing up, we get the running time of the algorithm:
\begin{align*}
	\reach + \reach \cdot \cycle = \bigO(\reach \cdot \cycle).
\end{align*}

\subsection*{Liveness Parameterized by Contributors}

We elaborate on the algorithm for $\LCL(\sizeC)$.
To this end, we show that the reachability algorithm for $\LCR(\sizeC)$ from \cite{Chini2018,Meyer2019} can be used to obtain the required interfaces.
We prove the correctness of this approach.
Finally, we discuss the complexity of the derived algorithm for $\LCL(\sizeC)$ in more detail.

\subparagraph*{Obtaining the Interfaces.}
We recall the fine-grained algorithm for $\LCR(\sizeC)$ presented in~\cite{Chini2018,Meyer2019}.
Given a leader contributor system $\lcs$ and final states $F \subseteq Q_L$, it decides in time $\bigO(\contributortimeReach)$ whether there is an initialized computation $c^0 \rightarrow^* c$ of $\lcs$ with $\pi_L(c) \in F$.
To this end, it computes a table $T$ with an entry $T[S] \subseteq Q_L \times \Domain$ for each $S \subseteq Q_C$.
The entry $T[S]$ contains all pairs $(q,a)$ that can be reached via a computation where the contributors discover the states depicted in the set $S$.

To formalize, we need the concept of incrementing computations.
These never delete states of the contributors.
A computation $\rho = c^0 \rightarrow c^1 \rightarrow \dots \rightarrow c^n$ is called \emph{incrementing} if $\proj{C}(c_i) \subseteq \proj{C}(c_{i+1})$ for each $i$.
We also write $c^0 \rightarrow_{\inc} c^n$.
The following lemma shows that the algorithm computes the interfaces for all incrementing prefixes.
\begin{lemma}
	\label{Lemma:ReachabilityContributorCorrectness}
	Let $I = (S,q,a) \in \IF$ be an interface.
	Then, there is an initialized computation $c^0 \rightarrow^*_\inc c$ with $I(c)$ if and only if $(q,a) \in T[S]$.
\end{lemma}
For proving the lemma, we first restate a result from \cite{Meyer2019} showing correctness of the reachability algorithm.
To this end, we introduce the notion of states of a computation.
Let \mbox{$\rho = c^0 \rightarrow c^1 \rightarrow \dots \rightarrow c^n$} be a computation.
The \emph{states} of $\rho$ is the set of contributor states appearing along the computation. 
These are captured in
\begin{align*}
	\States_C(\rho) = \bigcup_{i \in [1..n]} \proj{C}(c^i). 
\end{align*}
Now we can restate the result.
It shows correctness of the algorithm for $\LCR(\sizeC)$.
\begin{lemma}[\cite{Meyer2019}]
	\label{Lemma:CorrectnessReachability}
	Let $q \in Q_L$, $a \in \Domain$, and $S \subseteq Q_C$.
	There is an initialized computation $\rho = c^0 \rightarrow^* c$ with $\proj{L}(c) = q$, $\proj{\Domain}(c) = a$, and $S = \States_C(\rho)$ if and only if $(q,a) \in T[S]$.
\end{lemma}
Note that Lemma \ref{Lemma:ReachabilityContributorCorrectness} is slightly different.
It explicitly asks for an incrementing computation $\rho = c^0 \rightarrow^*_\inc c$ such that $c$ matches a given interface $I = (S,q,a)$.
To bridge the gap, we show that plain computations can always be mimicked by incrementing ones.
\begin{lemma}
	\label{Lemma:IncrementingComputation}
	There is an initialized computation $\rho = c^0 \rightarrow^* c$ if and only if there is an initialized incrementing computation $\rho^\inc = d^0 \rightarrow^*_\inc d$ with
	\begin{align*}
		\proj{L}(d) = \proj{L}(c), \proj{\Domain}(d) = \proj{\Domain}(c), \text{ and } \proj{C}(d) = \States_C(\rho).
	\end{align*}
\end{lemma}
\begin{proof}
	If an incrementing computation $\rho^\inc = d^0 \rightarrow^*_\inc d$ is given, we set $\rho = \rho^\inc$.
	The requirements on the projections are met.
	In particular, we have $\proj{C}(d) = \States_C(\rho)$ by the fact that $\rho$ is incrementing.
	
	For the other direction, let a computation $\rho = c^0 \rightarrow^* c$ be given.
	Assume, $\rho$ is not incrementing.
	Otherwise, we are done.
	There are configurations $c^i$ and $c^{i+1}$ in $\rho$ such that $\proj{C}(c^{i+1})$ does not contain $\proj{C}(c^i)$.
	This means, there is a state $p \in \proj{C}(c^i) \setminus \proj{C}(c^{i+1})$.
	This state gets lost by the transition $c^i \rightarrow c^{i+1}$, there is only one contributor $P$ with current state $p$ which does a transition to another state.
	
	We apply the copycat lemma to get an additional contributor $P^\cop$ that mimics $P$.
	It copies every move of $P$.
	Once $P^\cop$ reaches state $p$, it keeps staying in the state.
	With the new contributor, the state does not get deleted and is preserved throughout the computation.
	
	We introduce such an additional contributor for each state $p$ that is deleted along $\rho$.
	Hence, we obtain an incrementing computation $\rho^\inc = d^0 \rightarrow^*_\inc d$ with $\proj{C}(d) = \States_C(\rho)$.
	Leader and memory act the same way as before.
	We get $\proj{L}(d) = \proj{L}(c)$ and $\proj{\Domain}(d) = \proj{\Domain}(c)$.
\end{proof}
We combine Lemma \ref{Lemma:CorrectnessReachability} and Lemma \ref{Lemma:IncrementingComputation} to prove Lemma \ref{Lemma:ReachabilityContributorCorrectness}.
\begin{proof}
	Assume there is an initialized computation $\rho = c^0 \rightarrow^*_\inc c$ with $I(c)$.
	Then, we get that $\proj{L}(c) = q$, $\proj{\Domain}(c) = a$ and $\proj{C}(c) = S$.
	Since $\rho$ is incrementing, we get that $\proj{C}(c) = \States_C(\rho)$.
	Hence, by Lemma \ref{Lemma:CorrectnessReachability} we get that $(q,a) \in T[S]$.
	
	For the other direction, let $(q,a) \in T[S]$.
	By Lemma \ref{Lemma:CorrectnessReachability} we get a computation $\rho = c^0 \rightarrow^* c$ with $\proj{L}(c) = q$, $\proj{\Domain}(c) = a$, and $S = \States_C(\rho)$.
	Invoking Lemma \ref{Lemma:IncrementingComputation}, we obtain an incrementing computation $\rho^\inc = d^0 \rightarrow^*_\inc d$ with $I(d)$.
	This completes the proof.
\end{proof}
\subparagraph*{Correctness of the Approach.}
For applying Theorem \ref{Theorem:TimeLCL}, we need to show that the interfaces extracted from the reachability algorithm are indeed all interfaces that witness a prefix.
To this end, we show that restricting to incrementing prefixes is sound and complete.
\begin{lemma}
	\label{Lemma:IncrementingPrefix}
	Let $q \in Q_L$ and $a \in \Domain$.
	There is a finite initialized computation $c^0 \rightarrow^* c \rightarrow^+_\sat c$ with $\proj{L}(c) = q$ and $\proj{\Domain}(c) = a$ if and only if there is a finite initialized computation $d^0 \rightarrow_\inc^* d \rightarrow^+_\sat d$ with $\proj{L}(d) = q$ and $\proj{\Domain}(d) = a$.
\end{lemma}
\begin{proof}
	One direction is trivial.
	For the other direction, let a computation $c^0 \rightarrow^* c \rightarrow^+_\sat c$ with $\proj{L}(c) = q$ and $\proj{\Domain}(c) = a$ be given.
	Let the prefix $c^0 \rightarrow^* c$ be denoted by $\rho$.
	By Lemma~\ref{Lemma:IncrementingComputation} there is an incrementing computation $\rho^\inc : d^0 \rightarrow^* d$ such that $\proj{L}(d) = q$ and $\proj{\Domain}(d) = a$.
	Moreover, following the proof of Lemma~\ref{Lemma:IncrementingComputation}, we observe that the computation $\rho^\inc$ is obtained from $\rho$ by only adding contributors.
	This means that the cardinality in each state grows, we get that $\card_p(d) \geq \card_p(c)$ for all $p \in Q_C$.
	
	Now we simulate the cycle $c \rightarrow^+_\sat c$ on the larger configuration $d$.
	Leader and memory act as before.
	Whenever there is a contributor in a certain state $p$ acting in $c \rightarrow^+_\sat c$, we can provide it also from $d$ since $\card_p(d) \geq \card_p(c)$.
	Hence, we get a cycle $d \rightarrow^+_\sat d$.
	Note that saturatedness is preserved since $\proj{C}(d) \supseteq \proj{C}(c)$.
\end{proof}
Let interface $I = (S,q,a)$ witness the existence of a prefix $c^0 \rightarrow^* c$ which is part of a live computation $c^0 \rightarrow c \rightarrow^+_\sat c$. 
By Lemma \ref{Lemma:IncrementingPrefix}, there is a live computation \mbox{$d^0 \rightarrow^*_\inc d \rightarrow^+_\sat d$} with incrementing prefix.
Moreover, the incrementing prefix is witnessed by an interface \mbox{$I' = (S',q,a)$} with $I'(d)$.
Hence, we can consider $I'$ instead of $I$.
This means that the interfaces obtained from the reachability algorithm, namely the interfaces witnessing incrementing prefixes actually suffice.
With these interfaces we can already witness all prefixes.

\subparagraph*{Complexity of the Algorithm.}
Like stated in the proof of Theorem \ref{Theorem:TimeLCL}, the algorithm for $\LCL(\sizeC)$ first calls the reachability algorithm for $\LCR(\sizeC)$.
According to Theorem \ref{Theorem:LCRContributor}, this takes time $\bigO(\contributortimeReach)$.
The algorithm computes the table $T$ which contains all interfaces.
Then, we iterate over all interfaces $(S,q,a)$ with $(q,a) \in T[S]$ and $q \in F$.
Each of these interfaces is passed as an input to $\CYC$.
The algorithm stops if a cycle is found.

We iterate over at most $2^\sizeC \cdot \sizeL \cdot \sizeD$ many interfaces.
Since a single invocation of $\CYC$ takes time $\bigO(\fixpointtime)$, the time needed for the complete iteration is $\bigO(2^\sizeC \cdot \sizeL \cdot \sizeD^2 \cdot (\sizeC^2 \cdot \sizeD + \sizeL^2 \cdot \sizeD^3))$.
Adding up the time complexities, we obtain the result depicted in Corollary \ref{Corollary:LCLContributorTime}.

\section{Proofs of Section \ref{Section:ParameterizationLeaderDomain}}
\label{Section:AppendixParameterizationLeaderDomain}

We provide proofs and details for Section \ref{Section:ParameterizationLeaderDomain}.

\subparagraph{Leader Validity.}
The leader should visit the sequence of states in $w$ and reach the target state $q$ while reading the values in $\beta$ at the positions indicated by $\sigma$.
Formally, $x = (w, q, \sigma)$ is \emph{valid for the leader wrt. $\beta$} if $\abs{\beta} = \ord(x)$ and for all $a_i$ the following holds.
If $a_i\neq \bot$, the leader has a transition $(q_i ,\wt{a_i}, q_{i+1}) \in \delta_L$.
If $a_i=\bot$, we have one of the following: $q_i = q_{i+1}$ or $(q_i, \varepsilon,q_{i+1}) \in \delta_L$ or $(q_i,\rd{b},q_{i+1}) \in \delta_L$. 
(Notice here that we slightly vary in our definition from the main section i.e. we add an additional condition that $q_i = q_{i+1}$. 
This is not a necessary addition but only so that the proofs can be greatly simplified.)
Here, $b$ is a value in $\beta$ written before position $i$.
Formally, $b$ is from the set $S_\beta(i) = \Setcon{\beta_{\ell}}{\sigma(\ell) \leq i}$.
Note that $q_{n+1} = q$.
We use the predicate $\LValid_{\beta}(x) = \TRUE$ \mbox{to denote that $x$ is valid for the leader wrt. $\beta$.}

\subparagraph{Contributor Validity.}
It is the contributors' task to provide the first writes along $w$ in the order indicated by~$\sigma$.
Let $\alpha$ be a first-write sequence of length $t$ with $t <\ord(x)$.
Assume the first writes in $\alpha$ were already provided and there is a $(t+1)$-st first write that has to be provided next.
To define the expected behavior of the contributors, we make explicit the writes they can rely on.
These stem from the leader and from fellow contributors.
For the leader, given a $q\in Q_L$ and a set $\Gamma \subseteq \Domain$ of values available to the leader due to first writes of $\alpha$, we define $\Loop(q,\Gamma) = \Setcon{b}{q \xrightarrow{u.!b.v}_L q \wedge u.!b.v \in (\rd{\Gamma}\ \cup\ \wt{D})^*}$.
This set contains all memory values that the leader may write in a loop at state $q$ while reading $\Gamma$.
The values that can be written by the contributors at a certain position are given by $S_\alpha(i)$.
With this, we obtain the regular language of writes available to the contributors:
\begin{align*}
\Expr(x, \alpha) = \Gamma_1^* \setcon{a_1,\varepsilon} \Gamma_2^* \setcon{a_2,\varepsilon} \dots \Gamma_j^*, \; \text{where} \; \Gamma_i = \Loop(q_i,S_{\alpha}(i)) \cup S_{\alpha}(i).
\end{align*}
Here, $j = \sigma(t+1)$ is the index of the $(t+1)$-st first write.
Moreover, \mbox{we interpret $a_i = \bot$ as $\varepsilon$.}

The witness $x$ is valid for the contributors wrt. $\beta$ if $\abs{\beta} = \ord(x)$ and if each value $\beta_i$ can be written by a contributor. 
To be precise, before writing the value, the contributor is only allowed to read from $\Expr(x, \beta_1\dots \beta_{i-1})$. 
To make this formal, fix $i \in [1..\ord(x)]$ and let $Q_i\subseteq Q_C$ be the contributor states that can produce the first write, $Q_i = \Setcon{p}{\exists p' : p \xrightarrow{\wt{\beta_i}}_C p' }$.
The set $\Trace_C(Q_i) = \Setcon{w}{\exists p \in Q_i: q^0_C \xrightarrow{w}_C p}$ contains the transition sequences that lead to~$Q_i$.
Let $h:\Ops{\Domain} \rightarrow \Domain \cup \setcon{\varepsilon}$ be the homomorphism that only preserves reads, $h(\wt{b}) = \varepsilon$ and  $h(\rd{b}) = b$ for each $b\in D$. 
Then the witness $x$ is \emph{valid for the $i$-th first write of $\beta$} if
\begin{align*}
\Expr(x, \beta_1\ldots\beta_{i-1}) \cap h(\Trace_C(Q_i) ) \neq \emptyset.
\end{align*}
We use $\CValid^i_\beta(x) = \TRUE$ to indicate non-emptiness of the intersection.
If $x$ is valid for all first writes, we call $x$ \emph{valid for the contributors wrt. $\beta$}.
Formally, the conjunction $\CValid_\beta(x) = \bigwedge_{i \in [1..\ord(x)]} \CValid^i_\beta(x)$ has to evaluate to $\TRUE$.

\subsection*{Proof of Lemma \ref{Lemma:ValidityPolyTime}}
Validity with respect to the leader is simple to verify:
the witness describes a run of the leader the existence of which can be checked in polynomial time.
For validity with respect to the contributors one needs to test whether the intersection 
$\Expr((w,q,\sigma),\beta_1 \dots \beta_{i-1}) \cap h(\Trace_C(Q_i))$
is non-empty for each first write.
Clearly this can be done in polynomial time.

\subsection*{Proof of Lemma \ref{Lemma:Witnesses}}

Here we need to prove that there is a computation of the form $ c_0 \rightarrow^* c$ with $\proj{L}(c) = q $ iff there is a witness $z =  (w, q, \sigma)$ and a first write sequence $\beta$ such that $\LValid_\beta(z) = \TRUE$ and $\CValid_\beta(z) = \TRUE$.

We first prove the easy direction where we assume the computation of the form $ c_0 \rightarrow^* c$ with $\proj{L}(c) = q $ and prove the existence of the witness. Let the sequence of transitions that appear in the assumed computation be $\tau_1 \dots \tau_n$, notice that there can be both transitions of leader and contributor in the same. Firstly mark all the transitions that belong to the leader (say with a color red). We will construct later the required witness string from these marked transitions. Now for each $d \in \Domain$, perform the following. Mark each of the contributor transition of the form $\tau_i = p \xrightarrow{\wt{d}}p' $ with a color say yellow. Now retain the very first transition marked yellow and delete rest of them. Complete the process for each of the memory values $d \in \Domain$, if there are no contributor write transitions corresponding to a memory value, we continue with the next one. Finally delete all the other contributor transitions that are not marked, let the resulting sequence be $\pi =  \tau_{i_1} \dots \tau_{i_j}$. Further let the sequence of transitions marked yellow be $\tau_{i''_1} \dots \tau_{i''_k} $ and the sequence marked red be $\tau_{i'_1} \dots \tau_{i'_{j-k}} $. Notice that the sequence of transitions marked yellow will automatically provide us with the first write sequence, let the sequence be $\beta = d_1 \dots d_k$ [i.e. the sequence of memory values that appear in $\tau_{i''_1} \dots \tau_{i''_k}$, in that order]. 

Now, let $\sigma : [1..k] \mapsto [1..j-k]$ be given by $\forall \ell \in [1..k], \sigma(\ell) = i''_\ell -\ell $, i.e. it  simply maps each first write to the number of leader transitions that occurs before it. 

To construct the witness string, let $\tau_{i'_\ell} = (q_{\ell},a_{\ell},q_{{\ell+1}})$. The required witness string is given by $w = (q_1,x_1)\dots (q_{{j-k}},x_{j-k})$, where $x_i = d$ if $a_i = \wt{d}$ for some $d \in \Domain$ and $x_i = \bot$ otherwise.
Clearly $x=(w,q,\sigma)$ is the required witness, it is easy to check that $\LValid_\beta(x) = \TRUE$ and $\CValid_\beta(x) = \TRUE$ for the same.

For the other direction, we assume that there is a valid witness $x= (w, q, \sigma)$ with respect to  a first write sequence $\beta$ and show that there is a computation of the form $ c_0 \rightarrow^* c$ with $\proj{L}(c) = q $.
Let $w = (q_1,a_1)\cdots(q_n,a_n)$ and let $\beta = b_1 \dots  b_k$. Since the witness is given to be valid, we have that $\LValid_\beta(x) = \TRUE$ and $\CValid_\beta(x) = \TRUE$.

Since $\LValid_\beta(x) = \TRUE$, there is a  valid sequence of transitions $t =  \tau_1 \dots \tau_n$ such that $\tau_i = (q_i,\wt{a_i},q_{i+1})$ if $a_i \neq \bot$, otherwise $\tau_i = (q_i,x,q_{i+1})$, where $x = \rd{d}$ for some $d \in \Domain$ or $x = \epsilon$.

Further since  $\CValid_\beta(x) = \TRUE$, we have for each $i \in [1..k]$, we have that $\Expr(x, \beta[1..i-1]) \cap h(\Trace_C(Q_i) ) \neq \emptyset $. We recall that 	$\Expr(x, \alpha) = \Gamma_1^* \setcon{a_1,\varepsilon} \Gamma_2^* \setcon{a_2,\varepsilon} \dots \Gamma_j^*$, where $ \Gamma_i = \Loop(q_i,S_{\alpha}(i)) \cup S_{\alpha}(i)$. Now for each $i$, let $\gamma^i$ be the witness string in 
$\Expr(x, \beta[1..i-1]) \cap h(\Trace_C(Q_i) )$, these are the reads that the contributor will ever perform (here, $\beta[1..i] = b_1\ldots b_{i}$). Let $\tau(\gamma^i)$ be sequence of transitions in the contributor that generates such a witness string. We let the function $\pi$ to be a monotonic function that maps each letter occurring in the witness string $\gamma^i$ to the position in the expression i.e. $\forall j \in [1..|\gamma^i|]$, $\pi(i,j) = \ell$ if $\gamma^i[j] \in \Gamma_\ell \cup \{a_\ell \} $, clearly $\pi(i,j) < \sigma(i) $. Intuitively this corresponds to the positions where the contributor reads the required symbol. We will also classify the type of the symbols that occur in each $\gamma^i$ as being $\ld, \ct, \lp$ corresponding to whether they are read of a leader write/ contributor write or a write due to a loop.

We let $\lambda$ to be the function defined as $\lambda (i,j) = \ld $ if $\pi (i,j) = \ell$ and $\gamma^i[j] = a_\ell$, $\lambda (i,j) = \lp $ if $\pi (i,j) = \ell$ and $\gamma^i[j] \in  \Loop(q_\ell,S_{\beta[1..i-1]}(\ell))$ and $\lambda (i,j) = \ct $ otherwise. For each $\lambda (i,j) = \lp $, let $\Loop(i,j)$ be the sequence of transitions that forms a loop and produces $\gamma^i[j]$ i.e. it is the sequence of transitions that witnesses a run of the form $ q_\ell \xxrightarrow{u.\wt{\gamma^i[j]}.v}_L q_\ell $ such that $ u.\wt{\gamma^i[j]}.v \in (\rd{S_{\beta[1..i-1]}(\ell)}\ \cup\ \wt{D})^*$.

We now show how to extend our sequence of leader transitions $t$ to $t^*$. For each $i \in [1..k]$ and for each $j \in [1..|\gamma^i|]$, if $\lambda(i,j) = \lp$, then we insert in position before the transition corresponding to $\ell = \pi(i,j)$ in $t$ (i.e. the $\ell^{th}$ transition in $t$) the sequence of transitions $\Loop(i,j)$. We do this based on the order of $i$ (i.e. we first for it for $i=1$, then for $i = 2$ and so on). 
We will assume that the newly added transitions are colored blue and the original ones white, we will need these colors later to specify the invariant that we will maintain when constructing the run. Notice that $t^*$ can include transitions that reads a value from the memory.
For any $d \in \Domain$, let $\#_d(t^*)$ represent the number of transitions in $t^*$ that read the value $d$ from memory. Similarly, let $\#_d(\gamma^i)$ represent the number of times $d$ occurs as a contributor read in $\gamma^i$ (i.e. $\#_d(\gamma^i) = | \{ j \mid \gamma^i[j ]= d \wedge \lambda(i,j) =\ct \}| $ ). Let $\#_d = \#_d(t^*) + \#(\gamma^1) + \cdots \#(\gamma^k)$, this will be the number of contributors we will need for each  $d \in \Domain$ that contributor can write.

We now show how to construct the required run in the leader contributor system, for this we start with a configuration consisting of the leader initial state and corresponding to each $d \in \{b_1 \cdots b_k \}$, we have $\#_d$ many contributors in the initial contributor state. We will refer to these set of contributors collectively as $[d]$. 
The intention is to move them collectively [i.e. they make similar moves simultaneously till they reach a state from where they can produce the letter $d$ ].
To construct the required run, we have one handle into each of $t^*, \gamma^1, \dots , \gamma^k $ that stores the index into these string, let these set of handles be $\idx = (\idx_{t^*},\idx_{\gamma^1},\dots,\idx_{\gamma^k })$. 
We will sometimes omit the subscript when it is clear from the context.
These handles store the position in the respective strings to indicate the position up to which the string has been processed, initially they are set to the first location in the string. The run that we construct will have the property that for each $d \in \Domain$, there are at-least $\#_d(\idx) =  \#_d(t^*[\idx..]) + \#(\gamma^1[\idx.. ]) + \cdots \#(\gamma^k[\idx ..])$ many contributors (including the ones that are yet to reach a state from where $d$ can be written) that can still produce $d$ (here $\alpha[j ..]$ indicates the str $\alpha$ starting from $j$). 
Further we also maintain the invariant that for any $\idx$, if the number of white symbols in $t^*[1.. \idx]  = \ell$ (denoted $\white(t^*(\idx)) = \ell$), then  for each $j$ such that $\sigma(j) \leq \ell$, $\idx_{\sigma^j} = |\sigma^j|$ i.e. at the positions of first writes the corresponding contributors are available. This follows from the fact that $\pi(i,j) \leq \sigma(i)$. Finally we also maintain the invariant that the leader process is always in the target state of the transition $t^*(\idx)$.

We are now ready to construct the required run inductively. For the base case, we start with the initial configuration with $\#_d$ many contributors (for each $d \in \Domain$) in their initial state . Let $\rho$ be the run inductively constructed and let $\idx$ be the current index up to which we have processed. 
Firstly for each $i \in [1..k]$, we make any possible internal moves of $\tau(\gamma^i)$ starting from the last transition that was executed in this sequence (recall this is the sequence of moves that generated the witness string $\gamma^i$). Suppose for any $i \in [1..j]$, if $\lambda(i, \idx_{\gamma^i}) =\ct$ and $\gamma^i(\idx) = b_j $ for some $j \leq i$, then clearly $\sigma(j) \leq \white(t^*(\idx))$ [For any contributor read, the first write is always before]. 
From this and our invariant, we have that there are $\#_{b_j}(\idx)$ many contributor in the state that can produce $b_j$, we can send one contributor to write the required value to memory. Following this, we move all the contributors in $[b_j]$ to execute the corresponding transition in $\tau(\gamma^i)$, we also increment $\idx_{\alpha^i}$ (notice that this would ensure that our invariant is not violated). 
Suppose for some $i \in [1..k]$, we have that $\lambda(i,\idx) = \lp$ and $\pi(i,\idx) = \white(t^*(\idx))-1$, then clearly there is a loop sequence $\Loop(i,\white(t^*(\idx)))$ that is present. We execute such a sequence till the loop writes the required symbol onto shared memory, move the set of contributors $[b_i]$ to execute the corresponding read transition. We then execute the rest of transitions in the loop. Notice that executing the loop may require reading contributors, however existence of  contributors that can provide such symbols is ensured by our invariant. Finally we update the $\idx$ by moving $\idx_{t^*}$ to position at end of the loop and by incrementing $\idx_{\gamma^i}$. It is easy to see that even in this case the invariant is maintained. Also notice that we added loops so that the loop required by $\gamma^i$ is found earlier to $\gamma^j$ when $i < j$. Hence we can process each $\sigma^i$ completely before proceeding to the next one.

Finally we process the leader. If the current transition $t^*[\idx]$ is a read of the contributor, then we move one contributor to write the corresponding value to memory and make the leader move. We also update the $\idx$ by incrementing $\idx_{t^*}$. Otherwise we make the leader move and update the $\idx$. If the move of the leader was a write of value to shared memory, for each $i$ such that $\pi(i,\idx) = \white(t^*(\idx))$ and $\lambda(i,\idx) = \ld$, we execute the corresponding transition from $\tau(\gamma^i)$ which reads the value written by the leader and update $\idx$ appropriately. It is easy to see that such a run is the required valid run in the system.

\subparagraph*{Witness Concatenation.}
The \emph{witness concatenation} $(w_1,q_1,\sigma_1) \times  (w_2,q_2,\sigma_2) = (w_1.w_2, q_2, \sigma)$ concatenates the sequences of leader-memory pairs.
Note that this may repeat states.
The target state is the one of the second witness. 
The map $\sigma$ is given by $\sigma:[1..i+j] \mapsto [1.. |w_1| + |w_2|]$ with $\sigma(\ell) = \sigma_1(\ell)$ for all $\ell\leq i$ and $\sigma(\ell) = \sigma_2(\ell - i) + |w_1|$ for all $\ell \in [i+1 .. i+j]$.

\subparagraph*{Shrink Operator.}
Given a witness $(w, q, \sigma)$, the function $\Shrink$ removes the first repetition of states in $w$, if any.  
Let $w=(q_1,a_1)\ldots (q_n,a_n)$ and let $x$ be the least index such that $q_x = q_y$ for some $y \neq x$.
Fix the minimal of these $y$.
Then $\Shrink(w,q,\sigma)= (w',q,\sigma')$, where $w' = (q_1,a_1)\ldots (q_{x-1},a_{x-1})(q_y,a_y) \ldots (q_n,a_n)$. 
Moreover, $\sigma'(\ell) = \sigma(\ell)$ if $\sigma(\ell) < x$, $\sigma'(\ell) = x$ if \mbox{$x \leq\sigma(\ell) \leq y$} and $\sigma'(\ell) = \sigma(\ell) -y+x$ otherwise.
If the input is a short witness, $\Shrink$ is the identity.
We use $\Shrink^*$ for the repeated \mbox{application of $\Shrink$ until a fixed point is reached.}

\subsection*{Proof of Lemma \ref{Lemma:RestrictionToShortWitnesses}}

Before we turn to the proof of Lemma \ref{Lemma:RestrictionToShortWitnesses}, we prove some auxiliary statements that significantly simplify the proof.
The first lemma states that leader validity of a witness is preserved under repeatedly applying the shrinking operator.
\begin{lemma}
	\label{Lemma:LeaderValidShrink}
	Let $\beta$ be a first-write sequence and $x \in \Wit$ a witness with $\LValid_\beta(x) = \TRUE$. 
	Then, we have that $\LValid_\beta(\Shrink^*(x)) = \TRUE$.
\end{lemma}
\begin{proof}
	We show that $\LValid_\beta(\Shrink(x)) = \TRUE$.
	Then, the above statement follows by induction.
	To this end, assume $x$ is given by $(w,q,\sigma)$ with $w = (q_1,a_1) \dots (q_n,a_n)$.
	If $\Shrink(x) = x$, there is nothing to show.
	Otherwise, there are indices $r < t$ such that $\Shrink(x) = (w',q,\sigma')$ where $w' = (q_1,a_1) \dots (q_{r-1},a_{r-1}) . (q_t,a_t) \dots (q_n,a_n)$.
	The map $\sigma'$ of the witness is defined by $\sigma'(\ell) = \sigma(\ell)$ if $\sigma(\ell) < r$, $\sigma'(\ell) = r$ if $r \leq \sigma(\ell) \leq t$, and $\sigma'(\ell) = \sigma(\ell) - t + r$ for $\sigma(\ell) > t$.
	
	For proving leader validity, let $a_i$ be a symbol in $w'$.
	Since $a_i$ also occurs in $w$ and $\LValid_\beta(x) = \TRUE$, we get one of the following.
	(1) There is a write transition $q_i \xrightarrow{\wt{a_i}}_L q_{i+1}$, (2) $q_i = q_{i+1}$, (3) there is an $\varepsilon$-transition $q_i \xrightarrow{\varepsilon}_L q_{i+1}$, or (4) there is a read transition $q_i \xrightarrow{\rd{b}}_L q_{i+1}$ with $b \in S_\beta(i) = \Setcon{\beta_\ell}{\sigma(\ell) \leq i}$.
	
	For Cases (1) and (3), note that write and $\varepsilon$-transitions carry over from $x$ to $\Shrink(x)$.
	The only subtlety occurs when $i = r-1$.
	Validity of $x$ guarantees a transition $q_{r-1} \xrightarrow{\wt{a_{r-1}} / \varepsilon}_L q_r$.
	But $q_r = q_t$.
	Hence, we have the needed transition for $\Shrink(x)$.	
	
	In Case (2), we get that $q_i = q_{i+1}$. Since the operator $\Shrink$ cuts out the first occurrence of a repeating state, Case (2) can only happen when $i \geq t$.
	Then, the equality of states is also true in $\Shrink(x)$.
	
	In the last case, we have to show that the read transition carries over to $\Shrink(x)$.
	Essentially, we need to prove that the index shift that occurs when passing from $w$ to $w'$ is consistent with the sets $S_\beta(i)$ and $S'_\beta(i) = \Setcon{\beta_\ell}{\sigma'(\ell) \leq i}$.
	This means that the read symbol $b$ has to lie in the corresponding set $S'_\beta(i')$.
	To this end, we make precise the relations among the sets $S_\beta(j)$ and $S'_\beta(j)$ for each index $j \in [1..n]$.
	
	If $j \in [1..r-1]$, we immediately obtain
	\begin{align*}
		S'_\beta(j) = \Setcon{\beta_\ell}{\sigma'(\ell) \leq j} = \Setcon{\beta_\ell}{\sigma(\ell) \leq j} = S_\beta(j)
	\end{align*}
	from the definition of $\sigma'$.
	Hence, the sets are equal for indices strictly smaller than $r$.
	
	For $j \in [r..t]$, first note that $S_\beta(j) \subseteq S_\beta(t)$ since these sets grow monotonically.
	The latter set can be written as
	\begin{align*}
		S_\beta(t) = \Setcon{\beta_\ell}{\sigma(\ell) \leq t} = 
		\Setcon{\beta_\ell}{\sigma(\ell) < r \; \text{or} \; r \leq \sigma(\ell) \leq t} = 
		\Setcon{\beta_\ell}{\sigma'(\ell) < r \; \text{or} \; \sigma'(\ell) = r}.
	\end{align*}
	The last equivalence is due to the definition of $\sigma'$.
	Since $S'_\beta(r) = \Setcon{\beta_\ell}{\sigma'(\ell) \leq r}$ is equivalent to the last set occurring in the above equations, we obtain that $S_\beta(t) = S'_\beta(r)$ and hence, $S_\beta(j) \subseteq S'_\beta(r)$ for each $j \in [r..t]$.
	
	In the last case, $j$ is an index in $[t+1..n]$.
	Consider the following transformation steps:
	\begin{align*}
		S_\beta(j) &= \Setcon{\beta_\ell}{\sigma(\ell) \leq j} \\
		&= \Setcon{\beta_\ell}{\sigma(\ell) \leq t \; \text{or} \; t < \sigma(\ell) \leq j} \\
		&= S_\beta(t) \cup \Setcon{\beta_\ell}{t < \sigma(\ell) \leq j} \\
		&= S'_\beta(r) \cup \Setcon{\beta_\ell}{t < \sigma(\ell) \leq j}.
	\end{align*}
	Note that in the last step we used that $S_\beta(t) = S'_\beta(r)$.
	Now we find an equivalent description for the latter set in the union.
	For an index $\ell$ with $\sigma(\ell) > t$, we get by definition that $\sigma'(\ell) = \sigma(\ell) - t + r$.
	Hence, we have that $t < \sigma(\ell) \leq j$ if and only if $r < \sigma'(\ell) \leq j - t + r$.
	We can derive the following:
	\begin{align*}
		S'_\beta(r) \cup \Setcon{\beta_\ell}{t < \sigma(\ell) \leq j} &= S'_\beta(r) \cup \Setcon{\beta_\ell}{r < \sigma'(\ell) \leq j - t + r} = S'_\beta(j-t+r).
	\end{align*}
	Hence, $S_\beta(j) = S'_\beta(j - t + r)$.
	
	Assume, from Case (4) we get a transition $q_i \xrightarrow{\rd{b}}_L q_{i+1}$ with $b \in S_\beta(i)$.
	If $i \in [1..r-1]$, we obtain by the above discussion that $b \in S_\beta(i) = S'_\beta(i)$.
	If $i = t$, we obtain that $b \in S_\beta(t) = S'_\beta(r)$.
	In the last case, $i \in [t+1..n]$, we get that $b \in S_\beta(i) = S'_\beta(i-t+r)$.
	This proves leader validity of $\Shrink(x)$ and completes the proof.
\end{proof}
The following lemma extends the results from Lemma \ref{Lemma:LeaderValidShrink}.
It shows that shrinking operator , leader validity, and witness concatenation behave well with respect to each other.
Moreover, it provides a way to \emph{replace} a witness in a concatenation as long as leader validity is guaranteed.
\begin{lemma}
	\label{Lemma:LeaderValidConcatWellBehaved}
	Let $x = (w,q,\sigma)$ be a witness of order $k$ and $y$ a witness of order $p$ with $\init(y) = q$.
	Moreover, let $\beta = \beta_1 \dots \beta_{k+p}$ be a first-write sequence and $\LValid_\beta(x \times y) = \TRUE$.
	\begin{enumerate}[a)]
		\item \label{Lemma:LeaderValidConcatWellBehaved:PartA} We have $\LValid_\beta(x \times \Shrink^*(y)) = \TRUE$.
		\item \label{Lemma:LeaderValidConcatWellBehaved:PartB} Let $x' = (w',q,\sigma')$ be a witness of order $k$ and let $\beta' = \beta_1 \dots \beta_k$ be the prefix of $\beta$ of length $k$.
		If $\LValid_{\beta'}(x') = \TRUE$, then $\LValid_{\beta}(x' \times y) = \TRUE$.
	\end{enumerate}
\end{lemma}
\begin{proof}
	We first prove Part \ref{Lemma:LeaderValidConcatWellBehaved:PartA}).
	To this end, we fix some notation that is used throughout the proof.
	Let $w$, the word of witness $x$ be given by $w = (q_1,a_1) \dots (q_m,a_m)$.
	Let $y$ be the tuple $(v,p,\tau)$ where \mbox{$v = (q_{m+1},a_{m+1}) \dots (q_n,a_n)$} and $q_{k+1} = q$.
	Then, for the concatenation we get $x \times y = (w.v,p,\sigma.\tau)$.
	The map $\sigma.\tau$ maps the first writes as depicted in the definition of the concatenation:
	$\sigma.\tau(\ell) = \sigma(\ell)$ for $\ell \in [1..k]$ and $\sigma.\tau(\ell) = \tau(\ell - k) + m$ for $\ell \in [k+1 .. k+p]$.
	
	When applying the shrink operator to $y$, we get that $\Shrink(y) = (v',p,\tau')$.
	Assume that $\Shrink(y) \neq y$, otherwise there is nothing to prove.
	Then, there are indices $r < t$ such that $q_r = q_t$ and $v' = (q_{m+1},a_{m+1}) \dots (q_{r-1},a_{r-1}) . (q_t,a_t) \dots (q_n,a_n)$.
	A concatenation with $x$ therefore yields $x \times \Shrink(y) = (w.v',p,\sigma.\tau')$ with word 
	\begin{align*}
		w.v' = (q_1,a_1) \dots (q_m,a_m) \dots (q_{r-1},a_{r-1}) . (q_t,a_t) \dots (q_n,a_n) .
	\end{align*}
	and map $\sigma.\tau'$, defined similarly to $\sigma.\tau$.
	
	Now the reasoning is similar to Lemma \ref{Lemma:LeaderValidShrink}.
	We obtain the following relation among the sets $S_\beta(j) = \Setcon{\beta_\ell}{\sigma.\tau(\ell) \leq j}$ and $S'_\beta(j) = \Setcon{\beta_\ell}{\sigma.\tau'(\ell) \leq j}$.
	For $j \in [1..r-1]$, we have that $S_\beta(j) = S'_\beta(j)$.
	For $j \in [r..t]$, we get $S_\beta(j) \subseteq S'_\beta(r)$, and $S_\beta(t) = S'_\beta(r)$.
	Finally, if $j \in [t+1 .. n]$, we obtain $S_\beta(j) = S'_\beta(j - t + r)$.
	
	For leader validity, fix a symbol $a_i$ in $w.v'$.
	Since $\LValid_\beta(x \times y)$, there are four cases.
	(1) There is a write transition  $q_i \xrightarrow{\wt{a_i}}_L q_{i+1}$.
	This transition immediately carries over to the witness $x \times \Shrink(y)$.
	(2) The states $q_i$ and $q_{i+1}$ are equal.
	The equality of states is also true in $x \times \Shrink(y)$.
	(3) There is an $\varepsilon$-transition $q_i \xrightarrow{\varepsilon}_L q_{i+1}$ which also carries over.
	(4) There is a read transition $q_i \xrightarrow{\rd{b}}_L q_{i+1}$ with $b \in S_\beta(i)$.
	By the above considerations, $b$ also lies in the suitable set of first writes of the witness $x \times \Shrink(y)$.
	
	For the proof of Part \ref{Lemma:LeaderValidConcatWellBehaved:PartB}), we adjust the above notation.
	The witness $x = (w,q,\sigma)$ is given via the word $w = (q_1,a_1) \dots (q_m,a_m)$.
	Let $y = (v,p,\tau)$ with word $v = (s_1,b_1) \dots (s_n,b_n)$ and $x' = (w',q,\sigma')$ with $w' = (p_1,c_1) \dots (p_t,c_t)$.
	We consider the two concatenations $x \times y = (w.v,p,\sigma.\tau)$ and $x' \times y = (w'.v,p,\sigma'.\tau)$ with words
	\begin{align*}
		w.v &= (q_1,a_1) \dots (q_m,a_m) . (s_1,b_1) \dots (s_n,b_n), \\
		w'.v &= (p_1,c_1) \dots (p_t,c_t) . (s_1,b_1) \dots (s_n,b_n),
	\end{align*}
	and maps
	\begin{align*}
		\sigma.\tau(\ell) = 
		\left\lbrace
		\begin{aligned}
			\sigma(\ell), \; &\text{if} \; \ell \in [1..k], \\
			\tau(\ell), \; &\text{if} \; \ell \in [k+1..k+p],
		\end{aligned}
		\right.
		\; \;
		\sigma'.\tau(\ell) = 
		\left\lbrace
		\begin{aligned}
			\sigma'(\ell), \; &\text{if} \; \ell \in [1..k], \\
			\tau(\ell), \; &\text{if} \; \ell \in [k+1..k+p].
		\end{aligned}
		\right.
	\end{align*}
	To prove leader validity of $x' \times y$, pick a symbol in the word $w'.v$.
	Assume it is $c_i$ for an $i \in [1..t]$.
	By the assumption $\LValid_{\beta'}(x') = \TRUE$, we get that either there is a transition $p_i \xrightarrow{\wt{c_i} / \varepsilon}_L p_{i+1}$ or $p_i = p_{i+1}$ or there is a read transition $p_i \xrightarrow{\rd{b}} p_{i+1}$ for an $b \in S^{x'}_{\beta'}(i) = \Setcon{\beta_\ell \in \beta'}{\sigma'(\ell) \leq i}$.
	The first two cases immediately carry over to $x' \times y$.
	In the latter case, we need to show that $b$ lies in the correct set $S^{x' \times y}_\beta(i) = \Setcon{\beta_\ell}{\sigma'.\tau(\ell) \leq i}$.
	Recall that $i \leq t$ and that $\sigma'.\tau(\ell) \leq t$ if and only if $\sigma'.\tau(\ell) = \sigma'(\ell)$ by definition.
	But this means that $S^{x'}_{\beta'}(i) = S^{x' \times y}_\beta(i)$.
	Note that in the discussion, we also cover the special case $p_{t+1} = q = s_1$.
	
	Assume the picked symbol is $b_i$ for an $i \in [1..n]$.
	Since $\LValid_\beta(x \times y) = \TRUE$, we either get a transition $s_i \xrightarrow{\wt{b_i} / \rd{b} / \varepsilon}_L s_{i+1}$ or $s_i = s_{i+1}$ where $b \in S^{x \times y}_\beta(i + m) = \Setcon{\beta_\ell}{\sigma.\tau(\ell) \leq i+m}$.
	Note the index $i+m$ in the set of first writes.
	The simple cases carry over to $x' \times y$.
	In the case of a read transition, consider the following.
	\begin{align*}
		S^{x \times y}_\beta(i+m) = \Setcon{\beta_\ell}{\sigma.\tau(\ell) \leq i+m} = \setcon{\beta_1,\dots,\beta_k} \cup \Setcon{\beta_\ell}{m < \sigma.\tau(\ell) \leq i+m}.
	\end{align*}
	The last equation holds by the definition of $\sigma.\tau$.
	Moreover, we have that $\sigma.\tau(\ell) = \tau(\ell - k) + m$ if and only if $\sigma.\tau(\ell) > m$.
	And similarly, $\sigma'.\tau(\ell) = \tau(\ell - k) + t$ if and only if $\sigma'.\tau(\ell) > t$.
	Hence, we get the following chain of equalities.
	\begin{align*}
		S^{x \times y}_\beta(i+m)
		&= \setcon{\beta_1, \dots, \beta_k} \cup \Setcon{\beta_\ell}{m < \tau(\ell - k) + m \leq i + m} \\
		&= \setcon{\beta_1, \dots, \beta_k} \cup \Setcon{\beta_\ell}{t < \tau(\ell - k) + t \leq i + t} \\
		&= \setcon{\beta_1, \dots, \beta_k} \cup \Setcon{\beta_\ell}{t < \sigma'.\tau(\ell) \leq i + t} \\
		&= \Setcon{\beta_\ell}{\sigma'.\tau(\ell) \leq i+t} \\
		&= S^{x' \times y}_\beta(i+t).
	\end{align*}
	This shows that $b$ lies in the correct set $S^{x' \times y}_\beta(i+t)$ and completes the proof.
\end{proof}
The previous results can be used to show that short validity always implies leader validity.
\begin{lemma}
	\label{Lemma:ShortValidImpliesLeaderValid}
	Let $z$ be a short witness of order $k$ and $\beta = \beta_1 \dots \beta_k$ a fist-write sequence.
	If $\SValid_\beta(z) = \TRUE$, then we have that $\LValid_\beta(z) = \TRUE$.
\end{lemma}
\begin{proof}
	We prove the lemma by a case distinction.
	If $\ord(z) = 0$, we get by the definition of short validity that $\beta = \varepsilon$ and $\LValid_\varepsilon(z) = \SValid_\varepsilon(z) = \TRUE$.
	
	If $\ord(z) = k+ 1 > 0$ for a $k < \sizeD$ then, by the recursive definition of short validity, there are witnesses $x \in \Ord(k)$ and $y \in \Ord(1)$ such that $z = x \otimes y$ and $\LValid_\beta(x \times y) = \TRUE$.
	Since $z = \Shrink^*(x \times y)$, we get $\LValid_\beta(z) = \TRUE$ by an application of Lemma \ref{Lemma:LeaderValidShrink}.
\end{proof}
We use regular languages of the form $\Expr(x,\alpha)$ to make visible the writes that contributors can rely on when providing a next first write.
If all first writes of a sequence were already provided, the language slightly changes due to the availability of all first writes.
In this case, we speak of \emph{full expressions}.
The definition is as follows:

Let $x = (w,q,\sigma)$ be a witness with $w = (q_1,a_1) \dots (q_n, a_n)$ and $\beta$ a first-write sequence with $\abs{\beta} = \ord(x)$.
The \emph{full expression} of $x$ with respect to $\beta$ is the regular language
\begin{align*}
	\FullExpr(x,\beta) = \Gamma^*_1 \setcon{a_1, \varepsilon} \Gamma^*_2 \setcon{a_2,\varepsilon} \dots \Gamma^*_n \setcon{a_n,\varepsilon}, \;
	\text{where} \; \Gamma_i = \Loop(q_i,S_\beta(i)) \cup S_\beta(i).
\end{align*}
The next lemma shows that full expressions are preserved under shrinking.
\begin{lemma} \label{Lemma:FullExpressionShrink}
	For a first-write sequence $\beta$ and a witness $x$ with $\LValid_\beta(x) = \TRUE$, we have
	\begin{align*}
		\FullExpr(x,\beta)  = \FullExpr(\Shrink^*(x),\beta).
	\end{align*}
\end{lemma} 
\begin{proof}
	We show that the full expressions are invariant under the shrink operator.
	Formally, we prove that $\FullExpr(x,\beta) = \FullExpr(\Shrink(x),\beta)$.
	Invariance of leader validity under shrinking is due to Lemma \ref{Lemma:LeaderValidShrink}.
	Hence, the lemma then follows by induction.

	Let $x = (w,q,\sigma)$ be the given witness with $w = (q_1,a_1) \dots (q_n,q_n)$.
	If $\Shrink(x) = x$, there is nothing to show.
	Otherwise, there exist indices $r < t$ with $q_r = q_t$ such that $\Shrink(x) = (w',q,\sigma')$ where $w' = (q_1,a_1) \dots (q_{r-1},a_{r-1}) . (q_t,a_t) \dots (q_n,a_n)$.
	The map $\sigma'$ is given by $\sigma'(\ell) = \sigma(\ell)$ if $\sigma(\ell) < r$, $\sigma'(\ell) = r$ if $r \leq \sigma(\ell) \leq t$, and $\sigma'(\ell) = \sigma(\ell) - t + r$ otherwise.
	
	Considering the full expression defined by the witness $x$, we obtain
	\begin{align*}
		\FullExpr(x,\beta) = \Gamma^*_1 . \setcon{a_1,\varepsilon} \dots
		\Gamma^*_{r-1} . \setcon{a_{r-1},\varepsilon} .
		\Gamma^*_r . \setcon{a_r,\varepsilon} \dots
		\Gamma^*_t . \setcon{a_t,\varepsilon} \dots
		\Gamma^*_n . \setcon{a_n,\varepsilon}
	\end{align*}
	where $\Gamma_i = \Loop(q_i,S_\beta(i)) \cup S_\beta(i)$.
	The full expression defined by $\Shrink(x)$ is given by
	\begin{align*}
		\FullExpr(\Shrink(x),\beta) = \Sigma^*_1 . \setcon{a_1,\varepsilon} \dots
		\Sigma^*_{r-1} . \setcon{a_{r-1},\varepsilon} .
		\Sigma^*_t . \setcon{a_t,\varepsilon} \dots
		\Sigma^*_n . \setcon{a_n,\varepsilon}.
	\end{align*}
	To describe $\Sigma_i$ we use the notation $S'_\beta(i) = \Setcon{\beta_\ell}{\sigma'(\ell) \leq i}$.
	Then, the sets are given by $\Sigma_i = \Loop(q_i,S_\beta'(i)) \cup S_\beta'(i)$ for $i \in [1..r-1]$, $\Sigma_t = \Loop(q_t,S_\beta'(r)) \cup S_\beta'(r)$, and for $i \in [t+1 .. n]$ we have $\Sigma_i = \Loop(q_i,S_\beta'(i-t+r)) \cup S'_\beta(i-t+r)$.
	Note that we need the case distinction for the sets $\Sigma_i$ due to the index shift that occurs when going from $x$ to $\Shrink(x)$.
	
	Now we show the equality of the full expressions.
	To this end, we split them into three parts and show equality of the single parts.
	We proceed in three steps.

	\subparagraph{Step 1:}
	We prove the following equation to be correct:
	\begin{align*}
		\Gamma^*_1 . \setcon{a_1,\varepsilon} \dots \Gamma^*_{r-1} . \setcon{a_{r-1}, \varepsilon} = \Sigma^*_1 . \setcon{a_1,\varepsilon} \dots \Sigma^*_{r-1} . \setcon{a_{r-1}, \varepsilon}.
	\end{align*}
	It is enough to show that $\Gamma_i = \Sigma_i$ for $i \in [1..r-1]$.
	We have seen in the proof of Lemma \ref{Lemma:LeaderValidShrink} that $S'_\beta(i) = S_\beta(i)$ for these indices $i$.
	Hence, we get the desired equality.
	
	\subparagraph{Step 2:}
	We show the middle parts of the expressions to be equal.
	Formally:
	\begin{align*}
		\Gamma^*_r . \setcon{a_r, \varepsilon} \dots \Gamma^*_t . \setcon{a_t,\varepsilon} = \Sigma^*_t . \setcon{a_t,\varepsilon}.
	\end{align*}
	From the proof of Lemma \ref{Lemma:LeaderValidShrink} we know that $S_\beta(t) = S'_\beta(r)$.
	Hence, we obtain the equation $\Sigma_t = \Loop(q_t,S'_\beta(r)) \cup S'_\beta(r) = \Loop(q_t,S_\beta(t)) \cup S_\beta(t) = \Gamma_t$.
	Taking the equivalence into account and dropping $a_t$, it is left to show that 
	\begin{align*}
		\Gamma^*_r . \setcon{a_r, \varepsilon} \dots \Gamma^*_t = \Gamma^*_t.
	\end{align*}
	One inclusion is immediate.
	For the other one, we show that $a_r, \dots, a_{t-1}$ are contained in $\Gamma_t$ and that $\Gamma_r, \dots, \Gamma_{t-1}$ are actually subsets of $\Gamma_t$.
	
	Due to validity of $x$ with respect to the leader, $\LValid_\beta(x) = \TRUE$, we get a run $\rho$ on the leader $P_L$ of the form
	\begin{align*}
		q_t = q_r \xrightarrow{\wt{a_r} / \bot}_L q_{r+1} \xrightarrow{\wt{a_{r+1}} / \bot}_L \dots \xrightarrow{\wt{a_{t-1} / \bot}}_L q_t,
	\end{align*}
	where $q_i \xrightarrow{\bot}_L q_{i+1}$ denotes either a read of a symbol $b \in S_\beta(i)$ or an $\varepsilon$-transition.
	Since $S_\beta(i) \subseteq S_\beta(t)$ for each $i \in [r..t-1]$, all reads along $\rho$ are only from the set $S_\beta(t)$.
	This means that each $a_i$ with $i \in [r..t-1]$ is either $\bot$ or occurs as a write in a loop of $q_t$ where reads are restricted to the set $S_\beta(t)$.
	Phrased differently, $a_r, \dots, a_{t-1} \in \Loop(q_t,S_\beta(t)) \subseteq \Gamma_t$.
	
	Fix $i \in [r..t-1]$.
	We show that $\Gamma_i \subseteq \Gamma_t$.
	To this end, we reconsider the run $\rho$ from above and split it into two parts with middle $q_i$.
	We denote by $\rho_1$ the first part $q_t = q_r \rightarrow_L \dots \rightarrow_L q_i$.
	By $\rho_2$, we denote the latter part $q_i \rightarrow_L \dots \rightarrow_L q_t$.
	Let now $b \in \Gamma_i = \Loop(q_i,S_\beta(i)) \cup S_\beta(i)$.
	Then, either $b \in S_\beta(i) \subseteq S_\beta(t) \subseteq \Gamma_t$ or $b$ appears as a write on a loop in $q_i$ where reading is restricted to $S_\beta(i) \subseteq S_\beta(t)$.
	If $b$ appears as a write, we can append $\rho_1$ as prefix and $\rho_2$ as postfix to the corresponding run.
	Then, $b$ appears as a write in a loop in $q_t$ while reading is restricted to $S_\beta(t)$.
	Hence, $b \in \Loop(q_t,S_\beta(t)) \subseteq \Gamma_t$.
	
	\subparagraph{Step 3:} 
	We prove the equivalence of the latter parts of the expressions:
	\begin{align*}
		\Gamma^*_{t+1} . \setcon{a_{t+1},\varepsilon} \dots \Gamma^*_n . \setcon{a_n,\varepsilon} = \Sigma^*_{t+1} . \setcon{a_{t+1},\varepsilon} \dots \Sigma^*_n . \setcon{a_n, \varepsilon}.
	\end{align*}
	It suffices to show that $\Gamma_i = \Sigma_i$ for $i \in [t+1 .. n]$.
	To this end, let $i \in [t+1..n]$ be fixed.
	Like before, we refer to the proof of Lemma \ref{Lemma:LeaderValidShrink} and obtain $S_\beta(i) = S'_\beta(i-t+r)$.
	It yields
	\begin{align*}
		\Sigma_i = \Loop(q_i,S'_\beta(i-t+r)) \cup S'_\beta(i-t+r) = \Loop(q_i,S_\beta(i)) \cup S_\beta(i) = \Gamma_i.
	\end{align*}
	Altogether, the full expression is preserved under shrinking.
	This completes the proof.
\end{proof}
A further tool that we use in the proof of Lemma \ref{Lemma:RestrictionToShortWitnesses} is the blow up of witnesses.
It allows us to increase the order of a first-order witness.
Let $x = (w,q,\sigma)$ be a first-order witness.
Moreover, let $k \in \Naturals$ be a natural number such that $k < \sizeD$.
Then, we extend $x$ to a witness of order $k+1$ by mapping $k$ first writes to the first position and the remaining first write to the position indicated by $\sigma$.
The \emph{$(k+1)$-blow up} of $x$ is the witness $x^{(k+1)} = (w,q,\sigma^{(k+1)})$ where $\sigma^{(k+1)} : [1..k+1] \rightarrow [1..n]$ is given by
\begin{align*}
	\sigma^{(k+1)}(i) = 
	\left\lbrace
	\begin{aligned}
		1, \; \text{if} \; i \in [1..k], \\
		\sigma(1), \; \text{if} \; i = k+1.
	\end{aligned}
	\right.
\end{align*}
The following lemma states that the (full) expression of a product is the concatenation of the full expression of the left factor and the (full) expression of the blow up of the right factor.
\begin{lemma}
	\label{Lemma:ProductBlowUp}
	Let $x$ be a witness of order $k < \sizeD$ and $y$ a first-order witness.
	Moreover, let $\beta = \beta_1 \dots \beta_{k+1}$ be a first-write sequence and let $\beta'$ denote the prefix $\beta_1 \dots \beta_k$.
	Then we have
	\begin{enumerate}[a)]
		\item \label{Lemma:ProductBlowUpPartA} $\FullExpr(x \times y, \beta) = \FullExpr(x,\beta') . \FullExpr(y^{(k+1)},\beta)$,
		\item \label{Lemma:ProductBlowUpPartB} $\Expr(x \times y, \beta') = \FullExpr(x,\beta') . \Expr(y^{(k+1)},\beta')$.
	\end{enumerate}
\end{lemma}
\begin{proof}
	We first prove Part \ref{Lemma:ProductBlowUpPartA}).
	To this end, we let $x = (w,q,\sigma)$ with $w = (q_1,a_1) \dots (q_n,a_n)$ and $y = (v,p,\tau)$ with $v = (p_1,b_1) \dots (p_m,b_m)$ and $p_1 = q$.
	Consider the witness concatenation $x \times y = (w.v,p,\sigma.\tau)$.
	The full expression of it is given by 
	\begin{align*}
		\FullExpr(x \times y,\beta) = \Gamma^*_1 . \setcon{a_1,\varepsilon} \dots \Gamma^*_n . \setcon{a_n,\varepsilon} . \Sigma^*_1 . \setcon{b_1,\varepsilon} \dots \Sigma^*_m . \setcon{b_m,\varepsilon}.
	\end{align*}
	In the language, we have $\Gamma_i = \Loop(q_i,S^{x \times y}_\beta(i)) \cup S^{x \times y}_\beta(i)$ for each $i \in [1..n]$ and similarly $\Sigma_i = \Loop(p_i,S^{x \times y}_\beta(i+n)) \cup S^{x \times y}_\beta(i+n)$ for $i \in [1..m]$.
	
	Let $i \in [1..n]$.
	Then, by definition of $\sigma.\tau$, we obtain the following:
	\begin{align*}
		S^{x \times y}_\beta(i) = \Setcon{\beta_\ell}{\sigma.\tau(\ell) \leq i} = \Setcon{\beta_\ell \in \beta'}{\sigma(\ell) \leq i} = S^x_{\beta'}(i).
	\end{align*}
	This implies that $\Gamma_i = \Loop(q_i,S^x_{\beta'}(i)) \cup S^x_{\beta'}$ and hence we get the following equality:
	\begin{align*}
		\FullExpr(x,\beta') = \Gamma^*_1 . \setcon{a_1,\varepsilon} \dots \Gamma^*_n . \setcon{a_n,\varepsilon}.
	\end{align*}
	
	It is left to show that $\FullExpr(y^{(k+1)},\beta) = \Sigma^*_1 . \setcon{b_1,\varepsilon} \dots \Sigma^*_m . \setcon{b_m,\varepsilon}$.
	Let the blow up of $y$ be denoted by $y^{(k+1)} = (v,p,\tau^{(k+1)})$.
	Then, its full expression is given by
	\begin{align*}
		\FullExpr(y^{(k+1)},\beta) = L^*_1 . \setcon{b_1,\varepsilon} \dots L^*_m . \setcon{b_m,\varepsilon},
	\end{align*}
	where $L^*_i = \Loop(p_i,S^{(k+1)}_\beta(i)) \cup S^{(k+1)}_\beta(i)$ with $S^{(k+1)}_\beta(i) = \Setcon{\beta_\ell}{\tau^{(k+1)}(\ell) \leq i}$.
	We show that $L_i = \Sigma_i$ for each $i \in [1..m]$.
	To this end, it is enough to prove the equality of the first-write sets $S^{(k+1)}_\beta(i) = S^{x \times y}_\beta(i+n)$.
	
	By definition, we get the following for $i \in [1..m]$:
	\begin{align*}
		S^{(k+1)}_\beta(i) = \Setcon{\beta_\ell}{\tau^{(k+1)}(\ell) \leq i} = \setcon{\beta_1, \dots, \beta_k} \cup 
		\left\lbrace
		\begin{aligned}
			\setcon{\beta_{k+1}}, \; &\text{if} \; \tau(1) \leq i, \\
			\emptyset, \; &\text{otherwise}.
		\end{aligned}
		\right.
	\end{align*}
	By definition of the map $\sigma.\tau$, the sets $\setcon{\beta_1, \dots, \beta_k}$ and $\Setcon{\beta_\ell}{\sigma.\tau(\ell) \leq n}$ are equal.
	Hence, we can rewrite the above expression.
	Note that $\tau(1) > 0$.
	We obtain:
	\begin{align*}
		S^{(k+1)}_\beta(i) = \Setcon{\beta_\ell}{\sigma.\tau(\ell) \leq n} \cup
		\left\lbrace
		\begin{aligned}
			\setcon{\beta_{k+1}}, \; &\text{if} \; n < \tau((k+1) - k) + n \leq i + n, \\
			\emptyset, \; &\text{otherwise}.
		\end{aligned}
		\right.
	\end{align*}
	Then, by definition it follows
	\begin{align*}
		S^{(k+1)}_\beta(i) &= \Setcon{\beta_\ell}{\sigma.\tau(\ell) \leq n} \cup
		\left\lbrace
		\begin{aligned}
			\setcon{\beta_{k+1}}, \; &\text{if} \; n < \sigma.\tau(k+1) \leq i + n, \\
			\emptyset, \; &\text{otherwise}
		\end{aligned}
		\right. \\
		&= \Setcon{\beta_\ell}{\sigma.\tau(\ell) \leq i+n} \\
		&= S^{x \times y}_\beta(i+n).
	\end{align*}
	
	For Part \ref{Lemma:LeaderValidConcatWellBehaved:PartB}), consider the expression of $x \times y$
	\begin{align*}
		\Expr(x \times y,\beta') = \Gamma^*_1 . \setcon{a_1,\varepsilon} \dots \Gamma^*_n . \setcon{a_n,\varepsilon} . \Sigma^*_1 . \setcon{b_1,\varepsilon} \dots \Sigma^*_j,
	\end{align*}
	where $j + n = \sigma.\tau(k+1)$.
	Note that this implies $j = \tau(1)$.
	The sets $\Gamma_i$ and $\Sigma_i$ are given by $\Gamma_i = \Loop(q_i,S^{x \times y}_{\beta'}(i)) \cup S^{x \times y}_{\beta'}(i)$ for $i \in [1..n]$ and $\Sigma_i = \Loop(p_i,S^{x \times y}_{\beta'}(i+n)) \cup S^{x \times y}_{\beta'}(i+n)$ for $i \in [1..j]$.
	Note that the first writes refer to $\beta'$, we have $S^{x \times y}_{\beta'}(i) = \Setcon{\beta_\ell \in \beta'}{\sigma.\tau(\ell) \leq i}$.
	
	Let $i \in [1..n]$.
	Then we obtain from the definition of $\sigma.\tau$:
	\begin{align*}
		S^{x \times y}_{\beta'}(i) = \Setcon{\beta_\ell \in \beta'}{\sigma(\ell) \leq i} = S^x_{\beta'}(i).
	\end{align*}
	Similarly to the proof of Part \ref{Lemma:LeaderValidConcatWellBehaved:PartA}, we obtain $\FullExpr(x,\beta') = \Gamma^*_1 . \setcon{a_1,\varepsilon} \dots \Gamma^*_n . \setcon{a_n,\varepsilon}$.
	
	It is left to show that $\Expr(y^{(k+1)},\beta') = \Sigma^*_1 . \setcon{a_1,\varepsilon} \dots \Sigma^*_j$.
	By definition, we obtain
	\begin{align*}
		\Expr(y^{(k+1)},\beta') = L^*_1 . \setcon{b_1,\varepsilon} \dots L^*_{j'},
	\end{align*}
	where $j' = \tau^{(k+1)}(k+1) = \tau(1) = j$ and $L_i = \Loop(p_i,S^{(k+1)}_{\beta'}(i)) \cup S^{(k+1)}_{\beta'}(i)$.
	Now let $i \in [1..j]$.
	Since $\tau^{(k+1)}$ maps the first writes $\beta_1, \dots, \beta_k$ to position $1$, we obtain:
	\begin{align*}
		S^{(k+1)}_{\beta'}(i) = \Setcon{\beta_\ell \in \beta'}{\tau^{(k+1)}(\ell) \leq i} = \setcon{\beta_1, \dots, \beta_k}.
	\end{align*}
	The map $\sigma.\tau$ maps the first writes $\beta_1, \dots, \beta_k$ to positions smaller than $n$.
	Hence, we get
	\begin{align*}
		S^{x \times y}_{\beta'}(i+n) = \Setcon{\beta_\ell \in \beta'}{\sigma.\tau(\ell) \leq i+n} = \setcon{\beta_1, \dots, \beta_k} = S^{(k+1)}_{\beta'}(i).
	\end{align*}
	This implies $L_i = \Sigma_i$ and completes the proof.
\end{proof}
Under certain assumptions, shrinking operator and blow up commute.
The next lemma formalizes this observation.
The technical assumption that we have to make is that $\sigma$ maps the (only) first write to the first position in the word of the witness.
\begin{lemma}
	\label{Lemma:ShrinkingBlowUp}
	Let $x = (w,q,\sigma)$ be a first-order witness with $\sigma(1) = 1$.
	Let $y = \Shrink^*(x)$.
	For each $k < \sizeD$, we have the equality $\Shrink^*(x^{(k+1)}) = y^{(k+1)}$.
\end{lemma}
\begin{proof}
	The witness $y$ is obtained by shrinking $x$.
	Hence, we get that $y$ is of the form $y = (w',q,\sigma)$.
	Note that $\sigma$ will not change under shrinking since $\sigma(1) = 1$ is its only value.
	Now consider the blow up of $x$, $x^{(k+1)} = (w,q,\sigma^{(k+1)})$.
	Due to the definition of the blow up, $\sigma^{(k+1)}$ is the constant $1$-map.
	
	Shrinking $x^{(k+1)}$ will result in a short witness $\Shrink^*(x^{(k+1)}) = (w',q,\sigma^{(k+1)})$.
	Note that the word $w'$ coincides with the word of $y$.
	Moreover, $\sigma^{(k+1)}$ is preserved under shrinking since it is the constant $1$-map.
	If we blow up $y$, we get $y^{(k+1)} = (w',q,\sigma^{(k+1)})$.
	Hence, we obtain the desired equality which completes the proof.
\end{proof}
Finally, we need a lemma which transforms a witness into a similar witness that separates the last first write.
Technically, we need that the first-write map $\sigma$ is strictly increasing for the last element it maps.
The lemma is key to the induction step in the proof of Lemma \ref{Lemma:RestrictionToShortWitnesses}.
\begin{lemma}
	\label{Lemma:StrictlyIncreasingSigma}
	Let $x = (w,q,\sigma) \in \Wit$ be a witness of order $k+1$ with $k < \sizeD$ and $\beta$ a first-write sequence with $\LValid_\beta(x) \wedge \CValid_\beta(x) = \TRUE$.
	Then, we can construct a witness $\hat{x} = (\hat{w},q,\hat{\sigma})$ with $\init(\hat{x}) = \init(x)$ and $\LValid_\beta(\hat{x}) \wedge \CValid_\beta(\hat{x}) = \TRUE$ that satisfies 
	\begin{align*}
		\hat{\sigma}(i) < \hat{\sigma}(k+1) \; \text{for each} \; i \in [1..k].
	\end{align*}
\end{lemma}
\begin{proof}
	If $x$ already satisfies $\sigma(i) < \sigma(k+1)$ for any $i \in [1..k]$, we set $\hat{x} = x$.
	Otherwise, let $\sigma(k+1) = p$.
	We can write the word $w$ as follows:
	\begin{align*}
		w = (q_1,a_1) \dots (q_{p-1},a_{p-1}) . (q_p,a_p) \dots (q_n,a_n).
	\end{align*}
	The idea in the construction of $\hat{w}$ is to prolong the word $w$ by a copy of $q_p$ so that two different positions in $\hat{w}$ refer to the state.
	To this end, set
	\begin{align*}
	\hat{w} = (q_1,a_1) \dots (q_{p-1},a_{p-1}) . (q_p, \bot) . (q_p,a_p) \dots (q_n,a_n) .
	\end{align*}
	The map $\hat{\sigma}$ is defined by $\hat{\sigma}(i) = \sigma(i)$ for $i \in [1..k]$ and $\hat{\sigma}(k+1) = p+1$.
	Since $\sigma$ is monotonically increasing, we obtain the desired property $\hat{\sigma}(i) < \hat{\sigma}(k+1)$ from the definition.
	Moreover, $\hat{x}$ satisfies $\init(\hat{x}) = \init(x)$.
	It is left to show that $\hat{x}$ is valid for the leader and the contributors wrt. $\beta$.
	
	For the leader validity, we fist compare the the sets $S_\beta(j)$, associated to $x$, with $\hat{S}_\beta(j)$, associated to $\hat{x}$.
	Since we shift the index in the construction of $\hat{w}$, we will also get an index shift when moving from $S_\beta(j)$ to $\hat{S}_\beta(j)$.
	We reflect this in a case distinction.
	For the first case, let $j \in [1..p-1]$.
	Then we have that
	\begin{align*}
		\hat{S}_\beta(j) = 
		\Setcon{\beta_\ell}{\hat{\sigma}(\ell) \leq j} =
		\Setcon{\beta_\ell}{\sigma(\ell) \leq j} = 
		S_\beta(j).
	\end{align*}
	The equation comes from the fact that $\hat{\sigma}(\ell) = \sigma(\ell)$ if $\sigma(\ell) \leq j$ and $j \leq p-1$.
	
	For the case $j = p$, consider the following equivalence. It follows from $\hat{\sigma}(\ell) \leq p$ for each $\ell \in [1..k]$ and $\sigma(\ell) \leq p$ for any $\ell \in [1..k+1]$.
	\begin{align*}
		\hat{S}_\beta(p) = \setcon{\beta_1,\dots,\beta_k} = S_\beta(p) \setminus \setcon{\beta_{k+1}}.
	\end{align*}
	
	In the last case, let $j \in [p+1..n+1]$.
	Then, $\hat{\sigma}$ maps all the elements of $\beta$ to a position that is at most $j$.
	We have that $\hat{S}_\beta(j) = \setcon{\beta_1, \dots, \beta_{k+1}}$.
	The map $\sigma$ maps to positions that are strictly smaller than $j$, $S_\beta(j-1) = \setcon{\beta_1, \dots, \beta_{k+1}}$.
	Hence, $\hat{S}_\beta(j) = S_\beta(j-1)$.
	
	Now we prove the leader validity for all positions $j \in [1..n+1]$ along the same case distinction.
	Let $j \in [1..p-1]$.
	We have to show that there is a transition $q_i \xrightarrow{\wt{a_i} / \varepsilon / \rd{b}}_L$ with $b \in \hat{S}_\beta(j)$ or that $q_i = q_{i+1}$.
	By the leader validity of $x$ we get that either the states are equal or that there is a transition $q_i \xrightarrow{\wt{a_i} / \varepsilon / \rd{b}}_L q_{i+1}$ with $b \in S_\beta(j)$.
	Since $\hat{S}_\beta(j) = S_\beta(j)$ in that case, leader validity holds for position $j$.
	
	Consider the case $j = p$.
	By the definition of $\hat{w}$ we have that $q_p$ is the state of position $p$ and $p+1$.
	Hence, the states of the positions coincide and leader validity for position $p$ holds.
	
	For the last case, let $j \in [p+1..n+1]$.
	By $\LValid_\beta(x) = \TRUE$ we either get that $q_{j-1} = q_j$ or we obtain a transition $q_{j-1} \xrightarrow{\wt{a_{j-1} / \varepsilon / \rd{b}}}_L q_j$ with $b \in S_\beta(j-1) = \hat{S}_\beta(j)$.
	Hence, leader validity also holds in this case and we get that $\LValid_\beta(\hat{x}) = \TRUE$.
	
	Now we prove that $\CValid_\beta(\hat{x}) = \TRUE$.
	To this end, we show that the positions of the first writes within $\beta'$, a prefix of $\beta_1 \dots \beta_k$, under $\hat{\sigma}$ and $\sigma$ are the same.
	Let $j \in [1..p]$.
	Then
	\begin{align*}
		\hat{S}_{\beta'}(j) = 
		\Setcon{\beta_\ell \in \beta'}{\hat{\sigma}(\ell) \leq j} = 
		\Setcon{\beta_\ell \in \beta'}{\sigma(\ell) \leq j} =
		S_{\beta'}(j).
	\end{align*}
	Note that for the equality it is important to consider prefixes $\beta'$ which exclude the first write $\beta_{k+1}$.
	For $j \in [p+1..n+1]$ we have that $\hat{S}_{\beta'}(j) = \setcon{\beta_1, \dots, \beta_d} = S_{\beta'}(j-1)$ where $\beta_1 \dots \beta_d = \beta'$ denotes the considered prefix.
	
	Now we prove the equivalence of the expressions induced by $x,\hat{x}$ and $\beta$.
	Let $i \in [1..k]$ and $\beta' = \beta_1 \dots \beta_{i-1}$ a prefix.
	If we use the notation $\Sigma_j = \Loop(q_j,\hat{S}_{\beta'}(j)) \cup \hat{S}_{\beta'}(j)$ and \mbox{$\Gamma_j = \Loop(q_j,S_{\beta'}(j)) \cup S_{\beta'}(j)$}, we get the following two expressions:
	\begin{align*}
		\Expr(\hat{x},\beta') &= \Sigma_1^* . \setcon{a_1,\varepsilon} \dots \Sigma_{\hat{\sigma}(i)}^* , \\
		\Expr(x,\beta') &= \Gamma_1^* . \setcon{a_1,\varepsilon} \dots \Gamma_{\sigma(i)}^* .
	\end{align*}
	Since $i \leq k$, we get that $\hat{\sigma}(i) = \sigma(i)$ and $\hat{\sigma}(i) \leq p$.
	Thus, $\hat{S}_{\beta'}(j) = S_{\beta'}(j)$ for each $j \in [1..\hat{\sigma}(i)]$.
	This implies that $\Sigma_j = \Gamma_j$ and that the above expressions are the same.
	
	For $i = k+1$, the first-write sequence of interest is $\beta' = \beta_1 \dots \beta_k$.
	In this case, the expressions are of the form
	\begin{align*}
		\Expr(\hat{x},\beta') &= \Sigma_1^* . \setcon{a_1,\varepsilon} \dots \Sigma_p^* . \setcon{\varepsilon} . \Sigma_{p+1}^* , \\
		\Expr(x,\beta') &= \Gamma_1^* . \setcon{a_1,\varepsilon} \dots \Gamma_p^* .
	\end{align*}
	For $j \leq p$, we get that $\hat{S}_{\beta'}(j) = S_{\beta'}(j)$ by our earlier consideration.
	If $j = p+1$, we obtain $\hat{S}_{\beta'}(p+1) = S_{\beta'}(p)$.
	Hence, we get that $\Sigma_j = \Gamma_j$ for all $j \in [1..p]$ and $\Sigma_{p+1} = \Gamma_p$.
	Then the expressions again coincide.
	
	Since $\CValid_\beta(x) = \bigwedge_{i \in [1..k+1]} \CValid^i_\beta(x) = \TRUE$, we get that for each $i \in [1..k+1]$, the intersection $\Expr(x,\beta') \cap h(\Trace_C(Q_i))$ is non-empty, where $\beta' = \beta_1 \dots \beta_{i-1}$.
	Now we can replace $\Expr(x,\beta')$ by $\Expr(\hat{x},\beta')$ in each intersection and obtain that $\CValid^i_\beta(\hat{x}) = \TRUE$ for each $i \in [1..k+1]$ which implies $\CValid_\beta(\hat{x}) = \TRUE$.
\end{proof}
Finally, we turn to the proof of Lemma \ref{Lemma:RestrictionToShortWitnesses}.
\begin{proof}
	We fix a state $q \in Q_L$ and a first-write sequence $\beta$.
	For the first direction of the lemma, let a witness $x = (w,q,\sigma) \in \Wit$ with $\LValid_\beta(x) \wedge \CValid_\beta(x) = \TRUE$ be given.
	
	\subparagraph{First Direction:}
	By induction on the order of $x$, we prove a statement slightly stronger than depicted in the lemma.
	We show that there is a short witness $z = (w',q,\sigma')$ with $\init(z) = \init(x)$, $\ord(z) = \ord(x)$, $\FullExpr(z,\beta) = \FullExpr(x,\beta)$, and $\SValid_\beta(z) = \TRUE$.
	
	For the induction basis, consider the case where $\ord(x) = 0$.
	Then, $\beta = \varepsilon$.
	We set $z = \Shrink^*(x)$.
	Note that shrinking preserves initial state, target state, and order.
	Hence, the short witness $z$ is of the form $(w',q,\sigma')$ with $\init(z) = \init(x)$ and $\ord(z) = 0$.
	Recall that in this case, validity of $z$ is defined by $\SValid_\varepsilon(z) = \LValid_\varepsilon(z)$.
	Hence, we need to show validity of $z$ with respect to the leader.
	Since $\LValid_\varepsilon(x) = \TRUE$ by assumption, we obtain from Lemma \ref{Lemma:LeaderValidShrink} that $\LValid_\varepsilon(z) = \TRUE$.
	It is left to show that the full expressions of $z$ and $x$ coincide.
	But this follows immediately from Lemma \ref{Lemma:FullExpressionShrink}.
	
	Now assume that $\ord(x) = k+1$ for a $k \in \Naturals$ with $k < \sizeD$.
	Then, $\beta = \beta_1 \dots \beta_{k+1}$.
	We denote the prefix $\beta_1 \dots \beta_k$ of the first-write sequence by $\beta'$.
	Let $\sigma(k+1) = p$.
	Then, we can write the word $w$ as
	\begin{align*}
		w = (q_1,a_1) \dots (q_{p-1},a_{p-1}) . (q_p, a_p) \dots (q_n, a_n).
	\end{align*}
	By Lemma \ref{Lemma:StrictlyIncreasingSigma}, we can assume that $\sigma(i) < p$ for each $i \in [1..k]$.
	We define the word $w_\pre = (q_1,a_1) \dots (q_{p-1},a_{p-1})$ to be the prefix of $w$ up to the $(p-1)$-st letter.
	The remaining postfix is the denoted by $w_\po = (q_p,a_p) \dots (q_n,a_n)$.
	Moreover, we define the map $\sigma_\pre$ to be the restriction of $\sigma$ to $[1..k]$.
	Formally, $\sigma_\pre : [1..k] \rightarrow [1..p-1]$ with $\sigma_\pre(i) = \sigma(i)$.
	We further define $\sigma_\po$ to map a single first write to the first position $1$, $\sigma_\po(1) = 1$.
	Intuitively, $\sigma_\po$ is the map responsible for the last first write $\beta_{k+1}$.
	With these definitions we can split the witness $x$ into the following two witnesses
	\begin{align*}
		x_\pre = (w_\pre,q_p,\sigma_\pre) \; \text{and} \; x_\po = (w_\po,q,\sigma_\po).
	\end{align*}
	By definition, we get that $x = x_\pre \times x_\po$.
	Moreover, the orders are given by $\ord(x_\pre) = k$ and $\ord(x_\po) = 1$.
	We want to apply the induction hypothesis to $x_\pre$.
	To this end, we need to show that $\LValid_{\beta'}(x_\pre) \wedge \CValid_{\beta'}(x_\pre) = \TRUE$.
	
	For the leader validity, we use the fact that $\LValid_\beta(x) = \TRUE$.
	Let $j \in [1..p-1]$.
	By the leader validity of $x$, either $q_j = q_{j+1}$ or there exists a transition $q_j \xrightarrow{\wt{a_j} / \varepsilon / \rd{b}}_L$ with $b \in S_\beta(j)$.
	For the set $S_\beta(j)$, we have the following equivalence:
	\begin{align*}
		S_\beta(j) = 
		\Setcon{\beta_\ell \in \beta}{\sigma(\ell) \leq j} =
		\Setcon{\beta_\ell \in \beta'}{\sigma(\ell) \leq j} =
		\Setcon{\beta_\ell \in \beta'}{\sigma_\pre(\ell) \leq j} =
		S^\pre_{\beta'}(j).
	\end{align*}
	The first equality is by definition, the second by the fact that $j \leq p-1 < \sigma(k+1)$.
	The remaining equalities are again due to definition.
	Hence, $\LValid_{\beta'}(x_\pre) = \TRUE$.
	
	In order to see that $x_\pre$ is valid for the contributors wrt. to $\beta'$, consider the expressions induced by $x$ and $x_\pre$.
	Let $i \in [1..k]$.
	Since $S^\pre_{\beta'}(j) = S_\beta(j)$ for $j \in [1..p-1]$, we get 
	\begin{align*}
		\Expr(x_\pre,\beta_1 \dots \beta_{i-1}) = \Expr(x,\beta_1 \dots \beta_{i-1}).
	\end{align*}
	Hence, leader validity carries over to the witness $x_\pre$: $\CValid^i_{\beta'}(x_\pre) = \CValid^i_\beta(x) = \TRUE$.
	This means that also the conjunction of these values is true, $\CValid_{\beta'}(x_\pre) = \TRUE$.
	
	Now we can apply induction to $x_\pre$ and obtain a short witness $c = (w_c,q_p,\sigma_c) \in \Ord(k)$ with $\init(c) = q_1$, $\SValid_{\beta'}(c) = \TRUE$, and $\FullExpr(c,\beta') = \FullExpr(x_\pre,\beta')$.
	The witness $c$ is the first of two short witnesses that we will use in the recursion for short validity.
	The second witness is denoted by $d = (w_d,q,\sigma_d)$ and is defined by $d = \Shrink^*(x_\po)$.
	Then by definition, $d \in \Ord(1)$, $\init(d) = q_p$, and $\sigma_d(1) = 1$.
	Note that target state of $c$ and the initial state of $d$ match.
	Hence, the witness concatenation $c \times d$ is well-defined.
	
	The short witness of interest is then defined by $z = c \otimes d \in \Ord(k+1)$.
	Hence, $\ord(z) = \ord(x)$.
	Furthermore, we immediately get that $z$ is of the form $z = (w_z,q,\sigma_z)$ and that $\init(z) = \init(c) = q_1 = \init(x)$.
	It is therefore left to show that $z$ is valid, $\SValid_\beta(z) = \TRUE$, and that the full expressions coincide, $\FullExpr(z,\beta) = \FullExpr(x,\beta)$.
	
	We first focus on the validity of $z$.
	To this end, we make use of the recursive definition of $\SValid_\beta(z)$.
	It is enough to show that $\LValid_\beta(c \times d) = \TRUE$ and that $\CValid^{k+1}_\beta(c \times d) = \TRUE$.
	Note that $[z = c \otimes d]$ is true by definition and $\SValid_{\beta'}(c) = \TRUE$ holds by induction.
	
	Leader validity of $c \times d$ wrt. $\beta$ is obtained from the following chain of implications:
	\begin{align*}
		\LValid_\beta(x) \implies \LValid_\beta(x_\pre \times x_\po) \implies \LValid_\beta(c \times x_\po) \implies \LValid_\beta(c \times d).
	\end{align*}
	First note that $\LValid_\beta(x) = \TRUE$ by assumption.
	The first implication is due to the fact that $x = x_\pre \times x_\po$.
	For the second, we use that $\SValid_{\beta'}(c) = \TRUE$.
	We apply Lemma \ref{Lemma:ShortValidImpliesLeaderValid} and obtain that $\LValid_{\beta'}(c) = \TRUE$.
	Then, by Lemma \ref{Lemma:LeaderValidConcatWellBehaved}, we get that $\LValid_\beta(c \times x_\po) = \TRUE$.
	The last implication is again an application of Lemma \ref{Lemma:LeaderValidConcatWellBehaved} since $d = \Shrink^*(x_\po)$.
	
	Next, we show that $\CValid^{k+1}_\beta(c \times d) = \TRUE$.
	To this end, we prove 
	\begin{align*}
		\Expr(x,\beta') = \Expr(c \times d, \beta').
	\end{align*}
	Since $\CValid^{k+1}_\beta(x) = \TRUE$ by assumption, the equality of expressions implies that also $\CValid^{k+1}_\beta(c \times d)$ evaluates to $\TRUE$.
	Consider the expression of $c \times d = (w_c . w_d, q, \sigma_{c \times d})$ at $\beta'$.
	We have that 
	\begin{align*}
		\Expr(c \times d,\beta') = \FullExpr(c,\beta') . \Gamma^*_p,
	\end{align*}
	where $\Gamma_p = \Loop(q_p,S^{c \times d}_{\beta'} (\abs{w_c} + 1) ) \cup S^{c \times d}_{\beta'} (\abs{w_c} + 1)$. 
	The set of first writes $S^{c \times d}_{\beta'} (\abs{w_c} + 1)$ is given by $\Setcon{\beta_\ell \in \beta'}{\sigma_{c \times d}(\ell) \leq \abs{w_c} + 1}$.
	The equality holds since $\sigma_{c \times d}(k+1) = \abs{w_c} +1$, a fact that follows from $\sigma_d(1) = 1$.
	Since $\FullExpr(c,\beta') = \FullExpr(x_\pre,\beta')$, we get that 
	\begin{align*}
	\Expr(c \times d,\beta') = \FullExpr(x_\pre,\beta') . \Gamma^*_p.
	\end{align*}
	Now note that $S^{c \times d}_{\beta'} (\abs{w_c} + 1) = \setcon{\beta_1, \dots, \beta_k}$.
	This is due to $\sigma_{c \times d}(\ell) = \sigma_c(\ell) \leq \abs{w_c}$ for all $\ell \in [1..k]$.
	Moreover, we have the following equality of sets 
	\begin{align*}
			S^x_{\beta'} (p) = \Setcon{\beta_\ell \in \beta'}{\sigma(\ell) \leq p} = \setcon{\beta_1, \dots, \beta_k} = S^{c \times d}_{\beta'} (\abs{w_c} + 1).
	\end{align*}
	Hence, we obtain that $\Gamma_p = \Loop(q_p,S^x_{\beta'}(p)) \cup S^x_{\beta'}(p)$.
	Considering the expression of $x$ at $\beta'$, we then get the following
	\begin{align*}
		\Expr(x,\beta') = \Expr(x_\pre \times x_\po,\beta') = \FullExpr(x_\pre,\beta') . \Gamma^*_p
	\end{align*}
	since $\sigma(k+1) = p$.
	Thus, we have the desired equality.
	
	Finally, we prove that the full expressions of $z$ and $x$ coincide.
	To this end, we start with $\FullExpr(x,\beta)$ and transform it step by step to $\FullExpr(z,\beta)$.
	We begin with the following equalities which are consequences of $x = x_\pre \times x_\po$ and Lemma \ref{Lemma:ProductBlowUp}:
	\begin{align*}
		\FullExpr(x,\beta) = \FullExpr(x_\pre \times x_\po,\beta) = \FullExpr(x_\pre,\beta') . \FullExpr(x^{(k+1)}_\po,\beta).
	\end{align*}
	Since $d = \Shrink^*(x_\po)$ and $\sigma_\po(1) = 1$, we get by Lemma \ref{Lemma:ShrinkingBlowUp} that $d^{(k+1)} = \Shrink(x^{(k+1)}_\po)$.
	Hence, we obtain from Lemma \ref{Lemma:FullExpressionShrink} that $\FullExpr(x^{(k+1)}_\po,\beta') = \FullExpr(d^{(k+1)},\beta')$.
	Note that $x^{(k+1)}$ is leader valid wrt $\beta$ since $x$ is.
	Now we use that $\FullExpr(x_\pre,\beta') = \FullExpr(c,\beta')$ and get the equality:
	\begin{align*}
		\FullExpr(x_\pre,\beta') . \FullExpr(x^{(k+1)}_\po,\beta) = \FullExpr(c,\beta') . \FullExpr(d^{(k+1)},\beta).
	\end{align*}
	We apply Lemma \ref{Lemma:ProductBlowUp} and Lemma \ref{Lemma:FullExpressionShrink} again.
	Note that $z = \Shrink^*(c \times d)$ by definition.
	\begin{align*}
		\FullExpr(c,\beta') . \FullExpr(d^{(k+1)},\beta) = \FullExpr(c \times d,\beta) = \FullExpr(z,\beta).
	\end{align*}
	This completes the first direction of the proof.
	
	\subparagraph{Second Direction:}
	Now let a short witness $z = (w',q,\sigma')$ with $\SValid_\beta(z) = \TRUE$ be given.
	Like above, we employ induction to prove a slightly stronger statement.
	We show that there is a witness $x = (w,q,\sigma) \in \Wit$ with $\init(x) = \init(z)$, order $\ord(x) = \ord(z)$, $\FullExpr(x,\beta) = \FullExpr(z,\beta)$, and $\LValid_\beta(x) \wedge \CValid_\beta(x) = \TRUE$.
	
	For the induction basis, let $\ord(z) = 0$.
	In this case, $\beta = \varepsilon$.
	Set $x = z$.
	Then we only need to argue that $\LValid_\varepsilon(x) = \TRUE$ and $\CValid_\varepsilon(x) = \TRUE$.
	The latter holds since validity for contributors with empty first-write sequence is always true.
	Leader validity of $x$ holds since 
	\begin{align*}
		\LValid_\varepsilon(x) = \LValid_\varepsilon(z) = \SValid_\varepsilon(z) = \TRUE.
	\end{align*}
	
	Let $\ord(z) = k+1$ for $k < \sizeD$.
	Then, the first-write sequence is given by $\beta = \beta' . \beta_{k+1}$ with $\beta' = \beta_1 \dots \beta_k$.
	Since $\SValid_\beta(z) = \TRUE$, we get by the recursive definition of short validity, two witnesses $c \in \Ord(k)$ and $d \in \Ord(1)$ such that $z = c \otimes d$, $\LValid_\beta(c \times d) = \TRUE$, \mbox{$\CValid^{k+1}_\beta(c \times d) = \TRUE$}, and $\SValid_{\beta'}(c) = \TRUE$.
	We denote $c$ by $(w_c,q_c,\sigma_c)$ and $d$ similarly by $(w_d,q_d,\sigma_d)$.
	Note that $\init(c) = \init(z)$ and $q_d = q$.
	
	Since $c$ is a valid short witness of order $k$, we can apply induction.
	We obtain a witness $x' = (w_{x'},q_c,\sigma_{x'}) \in \Wit$ with initial state $\init(x') = \init(c) = \init(z)$, order $\ord(x') = k$, full expression $\FullExpr(x',\beta') = \FullExpr(c,\beta')$, and $\LValid_{\beta'}(x') \wedge \CValid_{\beta'}(x') = \TRUE$.
	The desired witness is $x = x' \times d$.
	Note that the concatenation is well-defined and that it immediately satisfies $x = (w,q,\sigma)$, $\init(x) = \init(z)$, and $\ord(x') = k + 1$.
	Hence, it is left to show that $\LValid_\beta(x) = \TRUE$, $\CValid_\beta(x) = \TRUE$, and that the full expressions of $x$ and $z$ coincide, $\FullExpr(x,\beta) = \FullExpr(z,\beta)$.
	
	We begin with leader validity.
	Since $\LValid_\beta(c \times d) = \TRUE$ and $\LValid_{\beta'}(x') = \TRUE$, we can apply Lemma \ref{Lemma:LeaderValidConcatWellBehaved}.
	It guarantees that $\LValid_\beta(x' \times d) = \TRUE$, which is what we wanted.
	
	For contributor validity, consider the following.
	We have seen that \mbox{$\CValid_{\beta'}(x') = \TRUE$} by induction.
	This means that each predicate $\CValid^i_{\beta'}(x')$ in the conjunction evaluates to $\TRUE$.
	We look at the corresponding expressions.
	For $x'$ and $x = x' \times d$, they are equivalent:
	\begin{align*}
		\Expr(x',\beta_1 \dots \beta_{i-1}) = \Expr(x, \beta_1 \dots \beta_{i-1})
	\end{align*}
	for each $i \in [1..k]$.
	The equation is due to $\sigma(i) = \sigma_{x'}(i)$ for $i \leq k$.
	Since $\CValid^i_{\beta'}(x') = \TRUE$, also the predicate $\CValid^i_{\beta}(x)$ evaluates to true for $i \in [1..k]$.
	It is left to argue that $\CValid^{k+1}_\beta(x) = \TRUE$.
	We make use of the fact that $\CValid^{k+1}_\beta(c \times d) = \TRUE$ and we show that the corresponding expressions of $x$ and $c \times d$ coincide.
	To this end, consider
	\begin{align*}
		\Expr(x,\beta') = \Expr(x' \times d,\beta') = \FullExpr(x',\beta') . \Expr(d^{(k+1)},\beta').
	\end{align*}
	The second equation follows by Lemma \ref{Lemma:ProductBlowUp}.
	Since the full expressions of $x'$ and $c$ coincide by induction, we get the following equations by invoking Lemma \ref{Lemma:ProductBlowUp} again:
	\begin{align*}
		\FullExpr(x',\beta') . \Expr(d^{(k+1)},\beta') = \FullExpr(c,\beta') . \Expr(d^{(k+1)},\beta') = \Expr(c \times d,\beta').
	\end{align*}
	This proves that the expressions are the same and that contributor validity carries over to $x$.
	We get $\CValid^{k+1}_\beta(x) = \TRUE$ and hence $\CValid_\beta(x) = \TRUE$.
	
	We show that the full expressions of $x$ and $z$ coincide.
	To this end, consider
	\begin{align*}
		\FullExpr(x,\beta) &= \FullExpr(x',\beta') . \FullExpr(d^{(k+1)},\beta) \\
		&= \FullExpr(c,\beta') . \FullExpr(d^{(k+1)},\beta) \\
		&= \FullExpr(c \times d,\beta) \\
		&= \FullExpr(z,\beta).
	\end{align*}
	The first and the third equation are due to Lemma \ref{Lemma:ProductBlowUp}.
	The second equation holds since the full expressions of $x'$ and $c$ are equivalent.
	Finally, the last equation is due to Lemma \ref{Lemma:FullExpressionShrink} which we can apply since $z = \Shrink^*(c \times d)$.
\end{proof}

\subsection*{Proof of Proposition \ref{Proposition:ValidWitnesses}}

It is left to explain the complexity.
Since there are $\bigO((\sizeL \sizeD)^\sizeL)$ many short witnesses and $\bigO(\sizeD^\sizeD)$ first-write sequences, the table has $\bigO((\sizeL \sizeD)^\sizeL \cdot \sizeD^\sizeD) = (\sizeL \cdot \sizeD)^{\bigO(\sizeL + \sizeD)}$ many entries.

To compute a single entry, we split $z$ into $x$ and $y$ by iterating over the short witnesses of order $k-1$ and $1$.
The iteration takes time proportional to the number of short witnesses $\bigO((\sizeL \sizeD)^\sizeL)$.
Checking whether $z = x \otimes y$ and evaluating $\LValid_\beta(x \times y) \wedge \CValid^{k+1}_\beta(x \times y)$ can be done in polynomial time.
Moreover, the value $\SValid_{\beta'}(x)$ can be looked up in the table.
Hence, computing an entry takes time $(\sizeL \sizeD)^{\bigO(\sizeL)}$.

The complete table, and hence all the values $\SValid_\beta(z)$, can thus be computed in time $(\sizeL \sizeD)^{\bigO(\sizeL + \sizeD)} \cdot (\sizeL \sizeD)^{\bigO(\sizeL)} = (\sizeL \sizeD)^{\bigO(\sizeL + \sizeD)} = (\sizeL + \sizeD)^{\bigO(\sizeL + \sizeD)}$.

\subsection*{Obtaining the Interfaces}

Let $z = (w,q,\sigma)$ with $w = (q_1,a_1) \dots (q_n,a_n)$ and $\beta$ a first-write sequence with $\SValid_\beta(z) = \TRUE$.
The state $q$ is the target state fixed by the witness.
The data value $a$ is the last symbol written in a computation along $z$.
It can either be $a_n$ or an arbitrary first write in $\beta$.
What remains is  to compute the set of all contributor states while conforming to the given short witness. We do this by iterating over all the contributor states and checking if it is reachable through the short witness. We start with an empty set of reachable contributors and will inductively build the required set by saturation. For each state of the contributor $c \in Q_C$, we check whether the contributor can reach the state $c$ from the initial state, when provided with the short witness as a support from the leader i.e. we check  $\Expr((w,q,\sigma),|\beta|) \cap h(\Trace_C(\{c\})) \neq \emptyset$. If the intersection is non empty then we add it to the set $S$. Iterating this procedure over all the states of contributor will give us the required set of reachable states $S$.

\section{Proofs of Section \ref{Section:Cycles}}
\label{Section:AppendixCycles}

We provide proofs and details for Section \ref{Section:Cycles}.

\subsection*{Proof of Lemma \ref{Lemma:Monotonicity}}

Let $\Gamma \subseteq \Gamma'$ be two subsets of $\Domain$.
Since the set of writes $\Writes(\decomp{S}{\Gamma})$ splits into $\Writes_C(\decomp{S}{\Gamma})$ and $\Writes_L(\decomp{S}{\Gamma})$, we show the two inclusions
\begin{align*}
	\Writes_C(\decomp{S}{\Gamma}) &\subseteq \Writes_C(\decomp{S}{\Gamma'}), ~\text{and} \\
	\Writes_L(\decomp{S}{\Gamma}) &\subseteq \Writes_L(\decomp{S}{\Gamma'}).
\end{align*}
To this end, let $\decomp{S}{\Gamma} = (S_1, \dots, S_\ell)$ and $\decomp{S}{\Gamma'} = (T_1, \dots, T_k)$ be the $\Gamma$-SCC decomposition and the $\Gamma'$-SCC decomposition of $S$.

For the first inclusion, take an element $b \in \Writes_C(S_1, \dots, S_\ell)$.
By definition, there are states $p,p'$ in a set $S_i$ and a transition $p \xrightarrow{\wt{b}}_C p'$.
Since $p,p'$ are in $S_i$, they are strongly connected in the graph $\restrictionGraph{S}{\Gamma}$.
Hence, the states are also strongly connected in $\restrictionGraph{S}{\Gamma'}$.
In fact, $\Gamma \subseteq \Gamma'$ implies that all the edges of $\restrictionGraph{S}{\Gamma}$ are also present in $\restrictionGraph{S}{\Gamma'}$.
Given that $(T_1, \dots T_k)$ is the $\Gamma'$-SCC decomposition of $S$, the states $p$ and $p'$ have to lie in one set $T_j$.
Hence, $b$ occurs as a write within a set of $(T_1, \dots, T_k)$ which means $b \in \Writes_C(T_1, \dots, T_k)$.

It is left to show the second inclusion.
Let $b \in \Writes_L(S_1, \dots, S_\ell)$.
Then, there are words $u,v \in \Ops{\Domain}^*$ such that $(q,a) \xrightarrow{u.\wt{b}.v}_{L'(\Gamma)} (q,a)$.
Recall that $\rightarrow_{L'(\Gamma)}$ is the transition relation of the automaton $P_{L'(\Gamma)}$.
It restricts the transitions of the leader to reads within the set $\Writes_C(S_1, \dots, S_\ell)$ and keeps track of the current memory content.
The latter may change due to a contributor write in $\Writes_C(S_1, \dots, S_\ell)$.
Since we already know that $\Writes_C(S_1, \dots, S_\ell) \subseteq \Writes_C(T_1, \dots, T_k)$, the automaton $P_{L'(\Gamma')}$ contains all the transitions of $P_{L'(\Gamma)}$.
Hence, the sequence of transitions $(q,a) \xrightarrow{u.\wt{b}.v}_{L'(\Gamma)} (q,a)$ in $P_{L'(\Gamma)}$ can also be carried out in $P_{L'(\Gamma')}$.
By definition, $b \in \Writes_L(T_1, \dots, T_k)$.

\subsection*{Proof of Proposition \ref{Proposition:CharacterizationStableSCCDecomp}}

We give an idea for proving the reverse direction.
A formal proof will be given afterwards.

Let $\Gamma$ be given.
We do not directly construct a saturated cycle, but a \emph{balanced} computation $\rho = c \rightarrow^+ d$ where $d$ and $c$ coincide up to the order of contributor states.
Phrased differently, $d$ is a permutation of $c$.
Moreover, $\rho$ is saturated in the above sense.
Since $d$ contains the same contributor states as $c$, $\rho$ can also be started in $d$.
This yields $c \rightarrow^+ d'$ where $d'$ is a new permutation of $c$.
Since there are only finitely many permutations, we eventually get a computation $c \rightarrow^* e \rightarrow^+_\sat e$ and hence, a saturated cycle.

Let $\decomp{S}{\Gamma} = (S_1, \dots, S_\ell)$.
To construct $\rho$, we first fix the behavior of the leader.
Formally, we pick a run $\rho_L$ of $P_{L'}$ from $(q,a)$ to $(q,a)$ that, on its way, writes all the symbols in $\Writes_L(S_1, \dots, S_\ell)$.
Note that such a run exists.
We let $t$ denote its length.
To execute $\rho_L$ properly, we have to provide the reads that it needs on the way.
Since these are from the set $\Writes_C(S_1, \dots, S_\ell)$, we construct supporting runs of the contributors providing them.

Let $b \in \Writes_C(S_1, \dots, S_\ell)$.
Then, there is a transition from $p$ to $p'$, both in $S_i$, writing $b$.
The idea is to keep enough copies of the \emph{source state} $p$ to provide $b$ whenever the leader needs it.
However, to obtain a balanced computation, we have to transfer the amount of contributors that moved from $p$ to $p'$ back to $p$.
Since $S_i$ is strongly $\Gamma$-connected, we know that there is a
path $p' \rightarrow^* p$ in $\restrictionGraph{S}{\Gamma}$.
Hence, there is a run on $P_C$ from $p'$ to $p$ reading only symbols from $\Gamma$.
With the above transition, we get a \mbox{cyclic run from $p$ to $p$.
We denote it by $\cyc(p)$.}

In the configuration $c$, we keep for each symbol $b$ with source state $p_b$ exactly $t+1$ copies of the states occurring in $\cyc(p_b)$.
We assume the contributors in $c$ are grouped into blocks $B_b(i)$ for $i \in [1..(t+1)]$.
Each block $B_b(i)$ simulates the run $\cyc(p_b)$.

When the leader starts to move along $\rho_L$, it might need to read a symbol $b$.
Then, there is a block $B_b(i)$ providing $b$.
To balance the block, all remaining transitions in it have to be executed.
Writes are simple.
They can be executed and ignored by other participants.
Read transitions in the block are handled in two different ways.

(1) Reads within the set $\Writes_C(S_1, \dots, S_\ell)$ are already executed in a special initial phase.
This explains the $(t+1)$-st copies of the cycles.
They are only used \mbox{to provide these reads.}

(2) Reads within $\Writes_L(S_1, \dots, S_\ell)$ are provided by the leader on $\rho_L$.
Since the leader traverses through all symbols in $\Writes_L(S_1, \dots, S_\ell)$, there is a transition which writes a particular symbol $b$ for the first time.
This write is then used to synchronize with all blocks.
The described computation is indeed balanced.
For more details, we refer to the formal proof.

\begin{proof}
	It remains to give a formal proof of the second direction.
	Let a non-empty set $\Gamma$ be given such that the $\Gamma$-SCC decomposition $\decomp{S}{\Gamma} = (S_1, \dots, S_\ell)$ is stable.
	This means that $\Gamma = \Writes(S_1, \dots,S_\ell)$.
	We split the set $\Gamma = \Gamma_C \cup \Gamma_L$, where $\Gamma_C = \Writes_C(S_1, \dots,S_\ell)$ are the writes of the contributors and $\Gamma_L = \Writes_L(S_1, \dots,S_\ell)$ are the writes of the leader.
	
	We fix a run of the leader.
	It is of the form $\pi = (q,a)\xrightarrow{w}_{L'} (q,a)$ and it writes every symbol in $\Gamma_L$.
	Formally, for each $g \in \Gamma_L$ there are $u,v \in \Ops{\Domain}^*$ such that $w = u \wt{g} v$. 
	Note that such a run exists.
	Potentially, we have to compose several cycles from $(q,a)$ to $(q,a)$.
	We denote the length of the run $\pi$ by $t$.
	
	For each element $b \in \Gamma_C$, let $p_0(b)$ and $p_1(b)$ be two states belonging to a set $S_{i(b)}$ of the $\Gamma$-SCC decomposition such that there is a transition $p_0(b) \xrightarrow{\wt{b}} p_1(b)$.
	Note that such a transition exists by definition.
	We call the set of states $\Gen = \Setcon{p_0(b)}{b \in \Gamma_C}$ the \emph{symbol generators}.
	Further, we fix a cycle for each symbol $b$.
	Let 
	\begin{align*}
		\cyc(b) = p_0(b) \rightarrow_C p_1(b) \rightarrow_C p_2(b) \rightarrow \dots \rightarrow_C p_k(b) = p_0(b)
	\end{align*}
	be a cyclic run in within $S_{i(b)}$, reading only symbols from $\Gamma$.
	Such a run exists since $S_{i(b)}$ is strongly connected in the graph $\restrictionGraph{S}{\Gamma}$.
	We use $\States(\cyc(b))$ to refer to the set $\setcon{p_0(b), \dots, p_{k-1}(b)}$ of states that appear in $\cyc(b)$.
	Moreover, given a configuration $c = (p,b,\pc)$ and a state $s$, we use $c[s]$ to denote the indices of the contributors that are currently in state $s$, $c[s] = \Setcon{j}{\pc(j) = s}$.
	
	We construct a computation $\rho$.
	The idea is to support the run $\pi$ of the leader and to provide all the needed symbols along its way.
	Moreover, we need to balance the computation:
	the number of contributors in a particular state is preserved after executing $\rho$.
	This is achieved by moving the contributors along the fixed cycles.
	
	For the construction, we start with $t+1$ many contributors in each state of $\cyc(b)$, for all symbols $b \in \Gamma_C$.
	Formally, we choose our initial configuration $c$ in such a way that for each $s \in S$ we have 
	\begin{align*}
		\abs{c[s]} = 
		\left\lbrace
		\begin{aligned}
			(t+1) \cdot \abs{ \Setcon{b \in \Gamma_C}{s \in \cyc(b)} }, ~&\text{if}~ s ~\text{lies in any cycle} \\
			1, ~&\text{otherwise.}
		\end{aligned}
		\right.
	\end{align*}
	Note that we add a single contributor in $s$ if the state does not appear in any cycle.
	This contributor does not move during the computation.
	The reason is that we can then ensure $\proj{C}(c) = S$ throughout the computation which keeps $\rho$ saturated.
	Moreover, we start with the appropriate leader state and memory value, $\proj{L}(c) =q,\proj{D}(c) = a$.
	
	During $\rho$, each contributor in a cycle moves to its neighbor by making exactly one move.
	To this end, we split $\rho$ into two phases: $\rho = \rho_1 . \rho_2$.
	In the first phase $\rho_1$, only the contributors move and the leader stays idle.
	The purpose of this phase is to ensure that all contributors can go to their neighbor in the cycle when reading a symbol from $\Gamma_C$ is required or when writing.
	Reading of other symbols is handled in $\rho_2$.
	
	Note that we have enough contributors in $c$ to provide each symbol in $\Gamma_C$ exactly $t+1$ many times.
	During $\rho_1$, we use up one of these contributors for each symbol and provide each symbol in $\Gamma_C$ once.
	To realize $\rho_1$, let $b \in \Gamma_C$.
	Pick one of the contributors currently in the state $p_0(b)$.
	It makes a move to $p_1(b)$ and writes $b$ to the memory.
	This is followed by a transition of every contributor in each of the cycles that can read $b$ and move to their neighbor.
	After the move, these contributors stay idle for the remainder of $\rho$.
	
	Let $c \rightarrow^* c'_1$ be the resulting computation.
	At the end of the computation, each transition in each copy of a cycle that involves reading a symbol from $\Gamma_C$ is already executed.
	Furthermore, one copy of the symbol generators is exhausted, the corresponding contributors made a move to the next state in the cycle.
	We still have $t$ contributors in the symbol generators left, $\abs{c'_1[p_0(b)]} = t$ for each $b \in \Gamma_C$.
	
	We complete the computation $\rho_1$.
	For any contributor in a state $s \in \cyc(b)$ that is not a symbol generator, $s \notin \Gen$, we do the following.
	If the contributor can write a symbol from $\Gamma_C$ and move to its neighbor state in $\cyc(b)$, we execute the transition.
	The written symbol is ignored by the other contributors and the leader.
	After executing these write transitions, we are at a configuration $\hat{c}_1$.
	We get $\rho_1 = c \rightarrow^* \hat{c}_1$.
	Still, we have $t$ contributors in the symbol generators left, $\abs{\hat{c}_1[p_0(b)]} = t$ for each $b \in \Gamma_C$.
	Hence, the contributors on the cycles that did not do a move so far are either the ones in the symbol generators or ones that require a symbol written by the leader, a symbol in $\Gamma_L$.
	
	We construct the second phase $\rho_2$ which shows how the leader runs.
	Recall that we already fixed the run $\pi$ of the leader providing all symbols in $\Gamma_L$.
	We execute each transition of $\pi$ interleaved with transitions of the contributors while maintaining two invariants.
	To formalize them, let $i \in [1..t]$.
	By $\Gamma^i_L \subseteq \Gamma_L$ we denote the set of symbols that the leader has written after $i$ many steps of $\pi$.
	The invariants are:
	(1) All contributors that are currently in a state $s \in \cyc(b)$ for a $b \in \Gamma_C$ but not in $\Gen$ and that can reach their neighbor while reading a symbol from $\Gamma^i_L$, have already performed this transition before the $(i+1)$-st step of $\pi$ is taken.
	(2) Before the $(i+1)$-st step of $\pi$ gets executed, for each $b \in \Gamma_C$, there are exactly $t-i$ many contributors left that can provide $b$.
	These are in the state $p_0(b)$.
	
	We construct the computation inductively.
	Assume, we already executed $i-1$ many steps of $\pi$.
	We denote the interleaved computation with the transitions of the contributors by $\rho^{i-1}_2$.
	We need a case distinction.
	
	If the $i$-th step of $\pi$, denoted by $\pi(i)$, is a write transition, we do not need to provide a symbol for the leader.
	The idea is to execute $\pi(i)$ and to let the contributors read the written symbol.
	Let $b \in \Gamma_L$ be that symbol.
	Then $\Gamma^i_L = \Gamma^{i-1}_L \cup \setcon{b}$.
	We first execute $\pi(i)$ and write $b$ to the shared memory.
	Now, each contributor on a cycle that needs to read a $b$ to arrive at its neighbor takes the corresponding read transition.
	This maintains Invariant (1).
	To ensure that (2) also holds, we add the following computation.
	For each symbol $b \in \Gamma_C$ we pick exactly one contributor in $p_0(b)$ and let it write $b$ to the memory.
	The write is ignored by others.
	This way, we consume exactly one copy of these contributors, maintaining (2).
	
	If $\pi(i)$ is a read of a symbol $b \in \Gamma_C$, we pick one contributor that is currently in $p_0(b)$.
	We let it execute its transition $p_0(b) \xrightarrow{\wt{b}} p_1(b)$ to provide $b$.
	The transition is followed by the leader taking $\pi(i)$.
	Invariant (1) is already ensured at this point since $\Gamma^i_L = \Gamma^{i+1}_L$.
	To guarantee (2), we consume copies for symbols different from $b$.
	Let $b' \in \Gamma_C$, $b' \neq b$.
	We let one copy of a contributor, currently in $p_0(b')$, perform its write transition on $b'$.
	The write is ignored by others.
	After executing these transitions, (2) holds. 
	
	Depending on the case, we add the resulting computation to $\rho^{i-1}_2$ and obtain a new computation $\rho^i_2$.
	Then we can define $\rho^2 = \rho^t_2$.
	Putting things together, we get
	\begin{align*}
		\rho = \rho_1 . \rho_2 = c \rightarrow^* \hat{c}_1 \rightarrow^* c_1.	
	\end{align*}
	By the maintained invariants, we get that $c_1$ is a permutation of $c$.
	All contributors took one transition along a cycle.
	Hence, the number of contributors in a certain state in $c$ and $c_1$ are equal.
	For each $s$ we have: $\abs{c[s]} = \abs{c_1[s]}$.
	Moreover, since $\pi$ is a cycle, we get $\proj{L}(c_1) = a = \proj{L}(c)$ and $\proj{\Domain}(c_1) = a = \proj{\Domain}(c)$.
	Hence, $\rho$ is a \emph{balanced} computation and can be applied again to $c_1$.
	
	Since there are only finitely many permutations of $c$, applying $\rho$ repeatedly will therefore yield a computation $c \rightarrow^* e \rightarrow^+_\sat e$ and hence, a saturated cycle.
\end{proof}

\subsection*{Proof of Lemma \ref{Lemma:OperatorProperties}}

We only need to show that for $X \subseteq \Domain$, the expression $\WritesSCC(X)$ can be evaluated in time $\bigO(\evaltime)$.
By definition, we have that $\WritesSCC(X) = \Writes(\decomp{S}{X})$.
We first compute $\decomp{S}{X}$.
To this end, we need to construct the graph $\restrictionGraph{S}{X}$.

To obtain $\restrictionGraph{S}{X}$, we iterate over the transitions in $\delta_C$.
If the current transition is a read within $X$ or a write, we keep it as an edge.
Hence, we need $\bigO(\abs{\delta_C}) = \bigO(\sizeC^2 \cdot \sizeD)$ time for the construction.
Note that a look-up in $X$ can be performed in constant time if we assume that $X$ is a bit-vector with $X(b) = 1$ if and only if $b \in X$.

Now we can apply Tarjan's algorithm to obtain the strongly connected components $(G_1, \dots, G_\ell)$ of $\restrictionGraph{S}{X}$.
Since the algorithm runs in time linear in the number of edges and the number of vertices, this takes time $\bigO(\sizeC + \abs{\delta_C}) = \bigO(\sizeC^2 \cdot \sizeD)$.
We obtain the $X$-SCC decomposition $\decomp{S}{X} = (S_1, \dots, S_\ell)$ by setting $S_i$ to the vertices of $G_i$.

It is left to compute the set $\Writes(S_1, \dots, S_\ell)$.
First, we focus on $\Writes_C(S_1, \dots, S_\ell)$.
To compute the set, we iterate over all transitions in $\delta_C$.
If the current transition is a write between two states $p,p'$ belonging to the same set $S_i$, we add the corresponding symbol to $\Writes_C(S_1, \dots, S_\ell)$.
We need $\bigO(\abs{\delta_C}) = \bigO(\sizeC^2 \cdot \sizeD)$ time for the iteration.
We can perform the check whether $p$ and $p'$ lie in the same set $S_i$ again in constant time.
Summing up, we needed $\bigO(\sizeC^2 \cdot \sizeD)$ time so far.

For computing $\Writes_L(S_1, \dots, S_\ell)$, we first need to construct the automaton $P_{L'}$.
The states $Q_L \times \Domain$ can be added in time $\bigO(\sizeL \cdot \sizeD)$.
The transitions of $P_{L'}$ are obtained by an iteration over $\delta_L$.
If the current transition is a write, $s \xrightarrow{\wt{b'}}_L s'$, then we add $\sizeD$ many transitions: $(s,b) \xrightarrow{\wt{b'}}_{L'} (s,b')$, one for each $b \in \Domain$.
If the transition is a read of a symbol $b$, we test whether $b \in \Writes_C(S_1, \dots, S_\ell)$ and add the single transition $(s,b) \xrightarrow{\rd{b}}_{L'} (s',b)$.
Adding these transitions takes time $\bigO(\abs{\delta_L} \cdot \sizeD) = \bigO(\sizeL^2 \cdot \sizeD^2)$ where the additional factor $\sizeD$ appears either since we add $\sizeD$ many transitions in the case of a write.
The $\varepsilon$-transitions in $P_{L'}$ can be added in time $\bigO(\sizeL \cdot \sizeD^2)$:
we iterate over each symbol $b' \in \Writes_C(S_1, \dots, S_\ell)$ and add $\sizeL \cdot \sizeD$ many transitions $(s,b) \xrightarrow{\varepsilon}_{L'} (s,b')$, one for each pair $(s,b)$.
Hence, we constructed the automaton $P_{L'}$ in time $\bigO(\sizeL^2 \cdot \sizeD^2)$.
Note that this limits the size of $\delta_{L'}$ to $\bigO(\sizeL^2 \cdot \sizeD^2)$.

To identify the elements in the set $\Writes_L(S_1, \dots, S_\ell)$, we iterate over all $b \in \Domain$ and test for each, whether it occurs as a write $\wt{b}$ on a cycle from $(q,a)$ to $(q,a)$ in $P_{L'}$.
The test can be reduced to a non-emptiness problem.
To this end, let $P_{L'}(q,a)$ be the automaton $P_{L'}$ with $(q,a)$ as initial and final state.
Then, $b \in \Writes_L(S_1, \dots, S_\ell)$ if and only if 
\begin{align*}
	\Ops{\Domain}^* . \wt{b} . \Ops{\Domain}^* \cap \langu(P_{L'}(q,a)) \neq \emptyset.
\end{align*}
Since the corresponding automaton for $\Ops{\Domain}^* . \wt{b} . \Ops{\Domain}^*$ has a constant number of states, building the product and deciding non-emptiness can be done in $\bigO(\abs{\delta_{L'}}) = \bigO(\sizeL^2 \cdot \sizeD^2)$ time.
Since the above non-emptiness test has to be executed for each $b \in \Domain$, we get a total time of $\bigO(\sizeL^2 \cdot \sizeD^3)$ to construct the set $\Writes_L(S_1, \dots, S_\ell)$.

Putting the sets $\Writes_C(S_1, \dots, S_\ell)$ and $\Writes_L(S_1, \dots, S_\ell)$ together, we obtain the complete set of writes, $\Writes(S_1, \dots, S_\ell) = \WritesSCC(X)$.
Adding up the complexities, we needed $\bigO(\evaltime)$ time for evaluating the operator.

\end{document}